\documentclass[a4paper]{article}

\RequirePackage{amsthm,amsmath}
\RequirePackage[numbers]{natbib}
\RequirePackage[colorlinks,citecolor=blue,urlcolor=blue]{hyperref}
\usepackage{amsmath}
\usepackage{amssymb}
\DeclareMathOperator*{\argomin}{arg\,min}
\DeclareMathOperator*{\argomax}{arg\,max}

\numberwithin{equation}{section}
\theoremstyle{plain}
\newtheorem{thm}{Theorem}[section]
\usepackage{verbatim}
\usepackage{graphicx}

\newcommand{\pr}[1]{P\hspace{-0.1em}\left(#1\right)}

\newcommand{\cPr}[2]{\mathbb{P}\hspace{-0.6mm}\left[\left.#1\,\right|#2\right]}

\newcommand{\E}[1]{\mathbb{E}\,#1}

\newcommand{\set}[1]{\left\{{#1}\right\}}
\newcommand{\eset}[2]{\left\{{#1} : \: {#2}\right\}}
\newcommand{\indic}[1]{1_{\{#1\}}}
\newcommand{\Nat}{\mathbb{N}}
\renewcommand{\Re}{\mathbb{R}}
\newcommand{\D}{{\,\mathrm{d}}}

\newtheorem{proposition}{Proposition}[section]
\newtheorem{definition}[proposition]{Definition}
\newtheorem{lem}[proposition]{Lemma}
\newtheorem{corollary}[proposition]{Corollary}

\newtheorem{remark}[proposition]{Remark}

\newtheorem{algorithm}[proposition]{Algorithm}
\newtheorem{construction}[proposition]{Construction}

\begin{document}

\title{Regular Decomposition: an information and graph theoretic approach to
 stochastic block models}

\author{Hannu Reittu\footnote{VTT Technical Research Centre of Finland Ltd, P.O.\ Box 1000, 02044 VTT. Email hannu.reittu@vtt.fi}, 
F\"ul\"op Bazs\'o\footnote{Department of Computational Science, Institute for Particle and Nuclear Physics, Wigner Research Centre for Physics, Hungarian Academy of Sciences, P.O.\ Box 49, H-1525 Budapest, Hungary. Email bazso.fulop@wigner.mta.hu.}, 
Ilkka Norros\footnote{University of Helsinki, Department of Mathematics and Statistics, P.O.\ Box 64, FI-00014 University of Helsinki. Email ilkka.norros@elisanet.fi}}
\date{}

\maketitle

\begin{abstract}
A method, regular decomposition (RD), for compression of large graphs and  matrices to
a block structure is proposed. Szemer\'edi's regularity lemma is used
as a generic motivation of the significance of the corresponding stochastic block
models (SBMs). Another ingredient of the method is Rissanen's minimum
description length principle (MDL). We analyze consistency of RD in detecting a block structure in a large and dense graph generated from a SBM with fixed number of blocks. We show that the coding length of the graph, used as a cost function in MDL, decreases until the right number of blocks is reached, then the coding length reaches a plateau with very slow or no reduction in value. This enables a practical algorithm for finding such block structures. Simulations are used to illustrate that this scenario is visible already in modest sizes of the graph. 
\end{abstract}

%\begin{keyword}[class=MSC]
%\kwd[Primary ]{68W40 }
%\kwd{05C50  }
%\kwd[; secondary ]{05C80 }
%\end{keyword}

%\begin{keyword}
%\kwd{Szemer\'edi's Regularity Lemma}
%\kwd{Minimum Description Length Principle}
%\kwd{Stochastic block models}
%\kwd{Graph compression}
%\end{keyword}


\tableofcontents

\section{Introduction}

In recent years, the analysis of large graphs, matrices and hypergraphs has become more contemporary. One reason for this is the elicitation of big data and its accumulation at an accelerating pace. 

Usually, the complete information of these objects (large graphs, matrices and hypergraphs) is not available or/and it's just too large to be practically applied. In this situation, the methods that can acquire a low-dimensional approximation of the underlying graphs are of a paramount interest. Such approximations can be seen as a compressed representation of the large-scale structure. This makes information theory a natural choice for the analysis which was pointed out by Rosvall and Bergstrom \cite{rosvall}. Roughly speaking, this means grouping redundant elements into a few communities. In case of graphs, examples are well-known division of nodes into communities or modules. Information theory will aid in choosing the optimal structures, in particular, the number of communities or modules. 

Our aim is to develop a corresponding rigorous methodology and proofs of correctness in the limit of large structures. We also introduce concrete algorithms and numerical examples that clarify how to use our method in practice. This work completes and partly summarizes our previous work on this subject \cite{tusnady,vesahannu,reittuweissbazso,reittuetall,reittuetalljournal,pirkko,reittudistance}.    

Another source of ideas is the fundamental mathematics of large structures which indicates the directions from which some exact results can be expected. A fundamental result in graph theory that is highly typical for problems of big data and related graphs is Szemer\'{e}di's Regularity Lemma (SRL) \cite{szemeredi76}. SRL is a
fundamental result in graph theory. Roughly speaking, SRL states that
any large enough graph can be approximated arbitrarily well by nearly
regular, pseudo-random bipartite graphs, induced by a partition of the
node set into a bounded number of equal-sized sets. For many graph
problems, it suffices to study the problem on a corresponding random
structure, resulting in a much easier problem (see,
e.g., \cite{komlos}). SRL is fundamental also in theoretical computer
science, say, in showing the existence of polynomial-time
approximations for solving dense graph problems, and in characterizing the
class of so-called testable graph properties \cite{alofisnewsha06}.

Despite the impressive theoretical applications of SRL, it has had
only few applications to 'real-life' problems. The main reason might
be that SRL has extremely bad worst-case scenarios in the sense that the lower
bound of graph sizes, for which the partition claim holds with
reasonable accuracy and without a single exception, is enormous. Thus,
a real-world application of SRL in the literal sense is
impossible. However, this drawback does not mean that regular
partitions could not appear in much smaller scales which are  relevant to
applications. On the contrary, one could conjecture that regular
partitions or structures be commonplace and worth revealing. 

The goal of our work is on realistic, yet preferably large networks
that appear in almost all imaginable application areas. The
regular structure granted by SRL is then replaced by a probabilistic
model, substituting the regular bipartite components with truly
random bipartite graphs. This model class is well known as stochastic
block models \cite{hollaslei83} (SBM), and it has recently gained much
attention in research and popularity in practical applications like
community detection \cite{peixoto,heilelmas12,massoulie14}. However,
the fundamental nature of SRL suggests the (heuristic) conjecture that
stochastic block models present a very generic form of the
separation of structure and randomness in large real-world
systems. That is why we think that SRL and related results should be kept in mind in practical applications as a rich source of abstractions that can lead to new applications. 

SBM structuring of data has good practical properties. Methods, such as maximum
likelihood fitting, expectation maximization, simulated annealing and Monte Carlo Markov
Chain algorithms, can be used. 

There are also some examples of graph decomposition applications that have
been explicitly inspired by SRL. The practical contexts are varied:
brain cortex analysis \cite{tusnady}, image processing \cite{pelillo},
peer-to-peer network\cite{vesahannu}, analyzing the functional magnetic resonance (fMRI) data to depict functional connectivity of the brain  \cite{pappas} and a matrix of multiple time series \cite{reittuweissbazso}. In the last mentioned work, the method was generalized from graphs to arbitrary positive matrices using a Poissonian construction as an intermediary step.

Interestingly, the authors of the recent work \cite{trivedisar} define
a 'practical' variant of SRL by relaxing algorithmic SRL in a certain
way to make it more usable in machine-learning tasks, see also
\cite{pelillo2017revealing}. Bolla \cite{bollabook} has developed a
spectral approach for finding regular structures of graphs and
matrices. Our emphasis is more information theoretical by nature and
continues the works of \cite{tusnady,vesahannu,reittuweissbazso,reittuetall, reittuetalljournal}. It
would be very interesting to compare the methods of
\cite{trivedisar,bollabook} and those of this paper in depth.

The third main ingredient in this paper is Rissanen's Minimum
Description Length (MDL) principle (see \cite{mdlgrunwald}), according
to which the quality of a model should be measured by the length of
the bit string that it yields for the unique encoding of the data. In our
case, the data has the form of a graph or matrix, and we present the
stochastic block models as a modeling space in the sense of the MDL
theory. Within this modeling space, the model corresponding to minimum
code length encoding presents the optimal regular decomposition
(partition) of the data. Therefore we call this Szemer\'edi-motivated
and technically MDL-based approach {\em Regular Decomposition}.

By information theory, the optimal coding reveals as much redundancy
in the data as is possible from within the given modeling framework. The
regular structure has a high degree of redundancy: a regular pair is
an almost structureless subgraph in which almost all the nodes have similar
and uniform connectivity patterns. By definition, the MDL principle
should be able to discover regular structures. Note that this
principle presents a case of 'Occam's Razor', a general rule of
reasoning that has proved fruitful in all areas of science.

In non-hierarchic clustering tasks, it has been a major challenge to
select the 'right' size of a partition ($k$). Intuitively, the optimal
choice of $k$ will strike a balance between the simplest partition of
the data using a single cluster and the maximal partition assigning
each data point to its own cluster, and selecting something in
between. A popular device has been the Akaike information criterion
(AIC), which was applied also in \cite{tusnady}. The AIC simply adds
the 'number of model parameters' to the maximal log-likelihood and chooses
the model that minimizes the sum. Our MDL-based approach solves the
corresponding model selection task in a better founded way (see also \cite{peixotoMDL}).

The contributions of this paper are the following: (i) the linkage of
SRL, stochastic block models and the MDL principle, (ii) the unified
handling of graphs and matrices, (iii) the effective practical algorithms
for revealing regular structures in data, and (iv) Theorem
\ref{blomothm} that characterizes how accurately the MDL principle
identifies a stochastic block model.

The paper is structured as follows: Section \ref{basicsec} presents
the definitions of the main notions: SRL, stochastic block models
and MDL. The last topic is expanded in Section
\ref{mdltheorysec}. Section \ref{resultsec} presents our main technical
results, in particular the algorithms and Theorem \ref{blomothm}. The
proof of Theorem \ref{blomothm} is given in Section \ref{blomothmsec}.
It is structured into several propositions and makes strong use of
information-theoretic tools presented in Appendix
\ref{chernoffsec}. 

\subsection{Related work}

Let us enlist contributions that are most relevant to core results of our work. 

Peixoto suggests to use  MDL in the SBM context in \cite{peixotoMDL} which continues earlier work by Rosvall and Bergstrom \cite{rosvall} where this idea was first suggested. It has interesting estimates which indicate that there is an upper limit of the number of communities that can be detected, $k\sim \sqrt{n}$ as a function of number of nodes $n$. However, Peixoto does not consider the exact detection of communities with a bounded number $k$ of communities, that is the focus of the current work. Peixoto has also made several implementations of corresponding algorithms and made them available \cite{peixotoMDL}.  

Wang and Bickel \cite{wangbickel} used information theory for likelihood-based model selection for the SBM. Their conclusions are similar to ours. In addition they show validity of results also in a sparse case when the average degree grows in polylog (polynomial in $\log n$, $n$ is number of nodes) rate and in the case of degree-corrected block models. They use asymptotic distributions instead of the exact ones that we use. The algorithmic part also deviates from ours. They end up in a likelihood-based criterion, Bayesian information criterion (BIC), that is asymptotically consistent. Along with a term that corresponds to log likelihood there is a term proportional to $k^2n\log n$ with some tuning coefficient that has to be defined separately in every concrete case. In so defined BIC has a minimum at the right value of the $k$. Instead of such a term we prefer to use MDL model complexity that is not case sensitive.   

\section{Basics and definitions}\label{basicsec}

\subsection{Szemer\'edi's Regularity Lemma}

Consider simple graphs $G(V,E)$, where $V$ is the set of nodes
(vertices) and $E$ is the set of links (edges). The {\em link density}
of a non-empty node set $X\subseteq V$ is defined as
\begin{equation*}
d(X)=\frac{|e(X)|}{\binom{|X|}{2}},\quad\mbox{where}\quad 
e(X)=\eset{\set{v,w}\in E}{v,w\in X},
\end{equation*}
and $\left|\cdot\right|$ denotes the cardinality of a set. Similarly,
the link density between two disjoint non-empty node sets
$X,Y\subseteq V$ is defined as
\begin{equation*}
d(X,Y):= \frac{|e(X,Y)|}{\left|X\right|\left|Y\right|},
\quad\mbox{where}\quad 
e(X,Y)=\eset{\set{v,w}\in E}{v\in X,\ w\in Y}.
\end{equation*}

\begin{definition}
Let $\epsilon>0$. A pair of disjoint sets $A,B\subseteq V$ of $G(V,E)$ is
called $\epsilon$-regular, if for every $X\subseteq A$ and
$Y\subseteq B$, such that $\left|X\right|> \epsilon \left|A\right|$
and $\left|Y\right|> \epsilon \left|B\right|$ we have
\begin{equation*}
\left|d(X,Y)-d(A,B)\right|< \epsilon.
\end{equation*}
A partition $\xi=\set{V_0,V_1,V_2,\cdots, V_k}$ of $V$ into $k+1$
sets, where all except $V_0$ have equal cardinalities, is called
$\epsilon$-regular, iff all except at most $\epsilon k^2$ pairs are
$\epsilon$-regular and $|V_0| <\epsilon |V|$.
\end{definition}

\begin{thm} (Szemer\'{e}di's Regularity Lemma, \cite{szemeredi76})
For every $\epsilon>0$ and for any positive integer $m$,
there are positive integers $N(\epsilon,m)$ and $M(\epsilon, m)$, such
that for any graph $G(E,V)$ with $\left|V\right|\geq N(\epsilon,m)$ there is an
$\epsilon$-regular partition of $V$ into $k+1$ classes with $m\leq
k\leq M(\epsilon,m)$.
\end{thm}

Roughly, SRL states that nodes of any large enough graph can
be partitioned into a bounded number ($k$) of equal-sized sets and
into one small set in such a way that links between most pairs of sets look like those in a
random bipartite graphs, whose link probability equals the link
density between the pair.

The claim of SRL is significant for sufficiently dense graphs,
i.e., when the link density is higher than $\epsilon$.The
result can be modified to sparse graphs by multiplying
$\epsilon$ at the right-hand side of the regularity definition by the
link density of the entire graph \cite{scott}. 

It is remarkable that the regularity claim holds for {\em
  all} graphs starting from a lower bound for size that depends only
on $\epsilon$. However, it is also well-known that this dependence on
$\epsilon$ is of extremely bad kind: the known lower bound for the
graph size $N(\epsilon,m)$ is extremely large, like a tower of powers of 2:
\begin{equation*}
2^{2^{.^{.^{.^{2}}}}},
\end{equation*}
where the height of the tower is bounded above by $1/\epsilon^5$. Such
a number is too big to be considered in any applications. Thus, all
real-world networks fall into a 'grey area' with respect to SRL. 

As we stated in the Introduction, there has been attempts to use algorithmic versions of  SRL (ASRL), introduced by Alon et al,
\cite{alonduke}, also in practical applications. The large numbers like the upper bound of $N(\epsilon,m)$ is problematic.  Although ASRL has time complexity that is only polynomial
$O(n^{2.376\cdots})$, corresponding to the time for multiplying two
$n\times n$ binary matrices, it requires this enormous size
($n\ge N(\epsilon,m)$) of graph to be able to find a regular
partition. 

A considerable improvement was found in the recent work
\cite{fischermatsliah}, where the execution time is only linear in the
graph size using a randomized algorithm. From a practical point of view, this algorithm
works in a more realistic fashion than the original ASRL: for any graph, it either finds an
$\epsilon$-regular partition or concludes that such a partition does
not exists. Another randomized algorithm with the same feature was
suggested by Tao \cite{taoblog}.

In principle, such algorithms could
possibly be applicable for real-world graphs, although some prohibitively big 
upper-bounds of  the constants of the
algorithm are a problem that needs a solution. That is why, we conclude that our approach is also needed. 

\subsection{Stochastic block models}\label{blomodefsec}

The notion of an $\epsilon$-regular partition is purely
combinatorial. The stochastic model closest to this notion is the
following. 

\begin{definition}
\label{blomodef}
Let $V$ be a finite set and $\xi=\set{A_1,\ldots,A_k}$ a partition of
$V$. A {\em stochastic block model} is a random graph $G=(V,E)$ with
the following structure: 
\begin{itemize}
\item There is a symmetric $k\times k$ matrix $D=(d_{ij})_{i,j=1}^k$ of real
numbers $d_{ij}\in[0,1]$ satisfying the {\em irreducibility condition}
that no two rows are equal, i.e.,
\begin{equation}
\label{irreducibility}
\mbox{for all }i,\ j,\ i<j,\mbox{ there is }q_{ij}\in\set{1,\ldots,k}
\mbox{ such that }d_{iq_{ij}}\not=d_{jq_{ij}};
\end{equation}
\item For every pair $\set{v,w}$ of distinct nodes of $V$ such that
  $v\in A_i$, $w\in A_j$, let $e_{vw}=e_{wv}$ be a Bernoulli random
  variable with parameter $d_{ij}$, assuming that all $e_{vw}$'s are
  independent. The edges of $G$ are
\begin{equation*}
E=\eset{\set{v,w}}{v,w\in V,\ v\not=w,\ e_{vw}=1}.
\end{equation*}
\end{itemize}
\end{definition}

Note that the case of the trivial partition $\xi=\set{V}$ yields the
classical random graph with edge probability $d_{11}$. 

A graph sequence $G_n=(V_n,E_n)$ presenting copies of the same stochastic
block model in different sizes can, for definiteness, be constructed
as follows. 

\begin{construction}\label{graphseqconstr}
Let $\gamma_1,\ldots,\gamma_k$ be positive, distinct real numbers such that
$\sum_{i=1}^k\gamma_i=1$. Divide the interval $(0,1]$ into $k$ segments
\begin{equation*}
  I_1=(0,\gamma_1],\ I_2=(\gamma_1,\gamma_1+\gamma_2],
\ldots,I_k=\left(\sum_{i=1}^{k-1}\gamma_i,1\right],
\end{equation*}
and denote $\Gamma=\set{I_1,\ldots,I_k}$. 
For $n=1,2,\ldots$, let the vertices of $G_n$ be 
\begin{equation*}
V_n=\eset{\frac{i}{n}}{i\in\set{1,\ldots,n}}.
\end{equation*}
For each $n$, let $\xi_n$ be the partition of $V_n$ into the blocks
\begin{equation*}
  A^{(n)}_i=I_i\cap V_n,\quad i=1,\ldots,k.
\end{equation*}
\end{construction}

For small $n$, we may obtain several empty copies of the empty set
numbered as blocks. However, from some $n_0$ on, all blocks are
non-empty and $\xi_n=\set{A^{(n)}_1,\ldots,A^{(n)}_k}$ is a genuine
partition of $V_n$. We can then generate stochastic block models based on
$(V_n,\xi_n,D)$ according to Definition \ref{blomodef}.

\begin{remark}
\rm A slightly different kind of stochastic block model can be defined
by drawing first the sizes of blocks $A^{(n)}_i$ as independent {\it
  Poisson}$(\gamma_in)$ random variables and proceeding then with the
matrix $D$ as before. The additional level of randomness, regarding
the block sizes, is however of no interest in the present paper.
\end{remark}

Next, we define the notion of a Poissonian block model in complete analogy with Definition \ref{blomodef}. (This allows almost one-to-one transfer of the proofs in Section \ref{blomothmsec} to the Poissonian case.)

\begin{definition}
\label{poiblomodef}
Let $V$ be a finite set of vertices, $n=|V|$, and let
$\xi=\set{A_1,\ldots,A_k}$ be a partition of $V$. The {\em symmetric
  Poissonian block model} is a symmetric random $n\times n$ matrix
$E$ with the following structure:
\begin{itemize}
\item There is a symmetric $k\times k$ matrix $\Lambda=(\lambda_{ij})_{i,j=1}^k$
  of non-negative real numbers satisfying the {\em irreducibility
    condition} that no two rows are equal, i.e.,
\begin{equation}
\label{poiirreducibility}
\mbox{for all }i,\ j,\ i<j,\mbox{ there is }q_{ij}\in\set{1,\ldots,k}
\mbox{ such that }\lambda_{iq_{ij}}\not=\lambda_{jq_{ij}};
\end{equation}
\item For every unordered pair $\set{v,w}$ of distinct nodes of $V$
  such that $v\in A_i$, $w\in A_j$, let $e_{vw}=e_{wv}$ be a Poisson
  random variable with parameter $\lambda_{ij}$, assuming that all
  $e_{vw}$'s are independent. The matrix elements of $E$ are $e_{vw}$
  for $v\not=w$, and $e_{vv}=0$ for the diagonal elements.
\end{itemize}
\end{definition}

Thanks to the independence assumption, the sums $\sum_{u\in
  A}\sum_{v\in B}e_{uv}$ are Poisson distributed for any $A,B\in\xi$.

\begin{remark}
The rest of the technical contents of this paper focus on the simple
binary and Poissonian models of Definitions \ref{blomodef} and
\ref{poiblomodef}. However, the following extensions are straightforward:
\begin{itemize}
\item {\em bipartite graphs:} this is just a subset of simple graphs;
\item {\em $m\times n$ matrices with independent Poissonian elements:}
  a matrix can be seen as consisting of edge weights of a bipartite
  graph, where the parts are the index sets of the rows and columns of
  the matrix, respectively;
\item {\em directed graphs:} a directed graph can be presented as a
  bipartite graph consisting of two parts of equal size, presenting
  the input and output ports of each node.
\end{itemize}
\end{remark}

\section{MDL approach to stochastic block models}\label{mdltheorysec}
In this Section we describe some basic definitions and notations for applying standard MDL modeling approach to graphs and matrices.   
\subsection{The Minimum Description Length (MDL) principle}\label{mdldefsec}

The Minimum Description Length (MDL) Principle was introduced by Jorma
Rissanen, inspired by Kolmogorov's complexity theory, and an extensive
presentation can be found in Gr\"unwald's monography
\cite{mdlgrunwald}, see also \cite{rissanen98}. The basic idea is the
following: a set $\mathcal{D}$ of data is optimally explained by a
model $\mathcal{M}$, when the combined unique encoding of the (i) model
and (ii) the data as interpreted in this model is as concise as
possible. By encoding we mean here a mapping that specifies an object
uniquely.

The principle is best illustrated by our actual case, simple graphs. A
graph $G=(V,E)$ with $|V|=n$ can always be encoded as a binary string
of length $\binom{n}{2}=n(n-1)/2$, where each binary variable
corresponds to a node pair and a value 1 (resp.\ 0) indicates an edge
(resp.\ absense of an edge). Thus, the MDL of $G$ is always at most
$\binom{n}{2}$. However, $G$ may have a structure whose disclosure
would allow a much shorter description. Our heuristic postulate is
that in the case of graphs and similar objects a good {\em a priori}
class of models should be inferred from SRL, which points to
stochastic block models.

\begin{definition}
\label{modelspace-nk-def}
Denote by $\mathcal{M}_{n/k}$ the set of irreducible stochastic block models
$(V,\xi,D)$ with
\begin{itemize}
\item $|V|=n$,
\item $|\xi|=k$, and, denoting $\xi=\set{V_1,\ldots,V_k}$,
\item for $i,j\in\set{1,\ldots,k}$, 
  \begin{equation*}
  d_{ij}=\frac{h_{ij}}{|V_i||V_j|},\ h_{ij}\in\Nat,\quad
d_{ii}=\frac{h_{ii}}{\binom{|V_i|}{2}},\ h_{ii}\in\Nat.
  \end{equation*}
\end{itemize}
\end{definition} 

The condition in the last bullet entails that each modeling space
$\mathcal{M}_{n/k}$ is finite.

\begin{remark}
Without the irreducibility condition \eqref{irreducibility}, there
would not be a bijection between stochastic block models and their
parameterizations.
\end{remark}

The models in ${\cal M}_{n/k}$ are parameterized by
$\Theta_k=(\xi,D)$. A good model for a graph $G$ is the one that gives
maximal probability for $G$ and is called the maximum likelihood
model. We denote the parameter of this model
\begin{eqnarray}
\hat{\Theta}_k(G):= \argomax_{ \Theta_k\in {\cal M}_{n/k}}(P(G\mid\Theta_k)),
\end{eqnarray}
where $P(G\mid \Theta_k)$ denotes the probability that the
probabilistic model specified by $\Theta_k$ produces $G$. 

One part of likelihood optimization is trivial: when a partition $\xi$
is selected for a given graph $G$, the optimal link probabilities are
the empirical link densities:
\begin{eqnarray}
d_{ij}=\frac{|e(V_i,V_j)|}{|V_i|| V_j|},\
i\neq j, \quad
d_{ii}=\frac{|e(V_i)|}{\binom{|V_i|}{2}}.
\end{eqnarray}
Thus, the nontrivial part is to find the optimal partition for the given
graph. This is the focus of the next sections.

\subsection{Two-part MDL for simple graphs}\label{twoparttheosec}

Let us denote the set of all simple graphs with $n$ nodes as
\begin{equation*}
\Omega_n=\eset{G}{G=(V,E)\mbox{ is a graph},\ |V|=n}.
\end{equation*}
A prefix (binary) coding of a finite set $\Omega$ is an injective
mapping
\begin{eqnarray}
C:\Omega\rightarrow  \cup_{s\geq 1}\{0,1\}^s
\end{eqnarray}
such that no code is a prefix of another code. Recall the
following proposition from information theory (see,
e.g., \cite{coverthomas}):

\begin{thm}\label{kraftthm}
(Kraft's Inequality) For an $m$-element alphabet there exists a binary
  prefix coding scheme with code lengths $l_1,l_2,\cdots , l_m$ iff
  the code lengths satisfy: $\sum_{i=1,\cdots,m}2^{-l_i}\leq 1$.
\end{thm}

An important application of Theorem \ref{kraftthm} is the following:
if letters are drawn from an alphabet with probabilities
$p_1,p_2,\cdots, p_m$, then there exists a prefix coding with code
lengths $\lceil -\log p_1\rceil,\cdots,\lceil -\log p_m\rceil$, and
such a coding scheme is optimal in the sense that it minimizes the
expected code length (in this section, the logarithms are in base
2). In particular, any probability distribution $P$ on the graph space
$\Omega_n$ indicates that there exists a prefix coding that assigns
codes to elements of $G\in\Omega_n$ with lengths equal to $\lceil
-\log{P(\set{G})}\rceil$.

The code length $l(\cdot)$ is the number of binary digits in the code
of the corresponding graph.  In case of a large set $\Omega$, most
such codes are long and as a result the ceiling function can be
omitted, a case we assume in sequel. A good model results in good
compression, meaning that a graph can be described by much less bits
than there are elements in the adjacency matrix. An incompressible case
corresponds to the uniform distribution on $\Omega_n$ and results in code
length $-\log{(1/\mid\Omega\mid)}=\binom{n}{2}$, equivalent to
writing down all elements of the adjacency matrix.

For every graph $G$ from $\Omega_n$ and model $P$ we can associate an encoding
with code length distribution $-\log{P(\cdot\mid
  \hat{\Theta}_k(G))}$. However, this is not all, since in order to
be able to decode we must know what particular probabilistic model $P$
is used. This means that also $\hat{\Theta}_k(G)$ must be prefix
encoded, with some code-length $L(\hat{\Theta}_k(G))$. We end up
with the following description length:
\begin{eqnarray}
\label{twopartcode}
l(G)=\lceil-\log{P(G\mid \hat{\Theta}_k(G))}\rceil+L(\hat{\Theta}_k(G)).
\end{eqnarray}
Eq.\ \eqref{twopartcode} presents the so-called {\em two-part MDL},
\cite{mdlgrunwald}. In an asymptotic regime with $n\rightarrow
\infty$, we get an analytic expression of the refined MDL. A simple
way of estimating $L(\hat{\Theta}_k(G))$ is just to map injectively
every model in ${\cal M}_{n/k}$ to an integer and then encode integers
with $l^*(\mid {\cal M}_{n/k}(G)\mid)$ as an upper bound of
the code-length. Here
\begin{equation}
\label{rissanenl}
l^* (m)=\max(0,\log{(m}))+\max(0,\log\log(m))+\cdots,\ m\in\Nat,
\end{equation}
gives, as shown by Rissanen, the shortest length prefix coding for
integers (see \cite{mdlgrunwald, Rissanen83}). The size of the graph must also
be encoded with $l^* (n)$ bits (it is assumed that there is a way of
defining an upper bound of the models with given $n$). In this point, it is
necessary to assume that the modeling space is finite. This results in

\begin{proposition}
\label{twopartprop}
For any graph $G\in \Omega_n$, there exists a prefix coding with code-length
\begin{align*}
l(G)&=\lceil-\log P(G\mid \hat{\Theta}_k(G))\rceil+m,\\
m&\leq m_k:= l^*(n)+l^*\left(S2(n,k)\left(\binom{n-k+2}{2}+1\right)^{\binom{k}{2}+k}
  +1\right)+1,
\end{align*}
where $S2(n,k)$ is the Stirling number of the second kind.
\end{proposition}
\begin{proof} 
The expression in \eqref{twopartcode} corresponds to a concatenation
of two binary codes. The $m$-part is the length of a code for describing the parameters of the model (in the case of a non-unique maximum, we take,
say, the one with smallest number in the enumeration of all such
models). The corresponding code is called the {\em parametric
  code}. The parametric code uniquely encodes the model. To create such an encoding, we just enumerate all possible
models, given in Definition \ref{modelspace-nk-def}, and use the
integer to fix the model. The length of a prefix code corresponding to
an integer is the $l^*$-function computed for that integer, and we add 1
to handle the ceiling function. 

To obtain an upper bound for the parametric code length $m$, we find
an upper bound for the number of models in the modeling space. The
number of models is upper-bounded by the product of two integers. The
first is the number of partitions of an $n$-element set into $k$
non-empty sets (blocks), which equals $S2(n,k)$, and the second bounds
the number of different link density configurations per partition. We
can view the blocks of a partition as the nodes of a `reduced
multi-graph' (in a multi-graph, there can be several links between a
node pair, as well as self-loops). The range of multi-links is between
zero and $\binom{n-k+2}{2}$: if we consider a pair of blocks (or one
block internally), there can be at most $n-(k-2)$ nodes in such a pair
(in one set, slightly less), since there must be at least $k-2$ nodes
in the other blocks of the partition. Obviously, in such a subgraph of
$n-(k-2)$ nodes there can be at most $\binom{n-k+2}{2}$ links. Thus,
the number of values each multi-link can take is upper-bounded by
$\binom{n-k+2}{2}+1$. Since the number of node pairs in the reduced
multi-graph is $\binom{k}{2}+k$, we obtain the second multiplier in
the argument of $l^*$ in the proposition.

Finally, we show that the coding of the graph is prefix. We
concatenate both parts into one code that has the prescribed length
and put first the prefix code of the integer that defines the
parameters of the maximum likelihood model. When we start to decode
from the beginning of the entire code, we first obtain a code of an
integer, because we used a prefix coding for integers. At this stage
we are able to define the probabilistic model that was used to create
the other part of the code, corresponding to the probability
distribution $P(\cdot\mid \hat{\Theta}_k(G))$. Using this information
we can decode the graph $G$. It remains to show that the concatenated
code itself is prefix. Assume the opposite: some prefix of such a code
is prefix to some other similar code, say, the first code is a prefix
to the second one. However, the parametric code was prefix, so both
codes must correspond to the same model. Since the first two-part code
is a prefix to the second, they both share the same parametric part,
and the code for the graph of the first is a prefix of the second
one. But this is impossible, since the encoding for graphs within the
same model is prefix. This contradiction shows that the two-part
coding is prefix.
\end{proof}

%{\bf Remark:} when we use the upper-bound code length in Proposition
%1, and use the Kraft's inequality in reverse direction we get a defect
%distribution on $\Omega$, this shows that the code lengths are somehow
%sub-optimal.

Finally, we call
\begin{eqnarray}
{\cal M}_{n}:=\bigcup_{1\leq k\leq n}{\cal M}_{n/k}
\end{eqnarray}
the {\em full regular decomposition modeling space of $\Omega_n$}.

\subsection{Two-part MDL for matrices}\label{twoparttheosecmatrix}

In this section we consider input data in the form of a $n\times m$ matrix $A=(a_{ij})$ with non-negative entries. With such a matrix we associate a random bipartite multi-graph. The set of rows and the set of columns form a bipartition. Between row $i$ and column $j$ there is a random number of links that are distributed according to Poisson distribution with mean $a_{i,j}$. Such a model was introduced in \cite{norrosreittugraphprocess} and it has been used in various tasks in complex network analysis, see \cite{van2016random}. The aim of this model is to back up, heuristically,  a corresponding practical algorithm for regular decomposition of matrices. Our approach is closely related to but slightly different from the Poissonian block model. Assume that $A$ is used to generate random $n\times m$ matrices $X$ with independent integer-valued elements following {\it Poisson}($a_{ij}$) distributions. The target is to find a regular decomposition model that minimizes the expected description length of such random matrices. 
We propose the following modeling spaces:  

\begin{definition}
For integers $k_1$, $k_2$ from ranges $1\leq k_1\leq n$ and $1\leq k_2\leq m$, the parameters of a model $\Theta_{k_1,k_2}$ in the modeling space ${\cal M}_{k_1,k_2}$ for an integer matrix $X$ are partition of rows into $k_1$ non-empty sets $V=(V_1,\cdots V_{k_1})$ and partition of columns into $k_2$ non-empty sets $U=(U_1,\cdots,U_{k_2})$ and $k_1\times k_2$ block average matrix $P$, with elements $(P)_{\alpha,\beta}:= \sum_{i\in V_\alpha, j\in U_\beta}\frac{x_{i,j}}{|V_\alpha||U_\beta|}$.
\end{definition}

Thanks to the addition rule of Poisson distributions, the likelihood of $X$ in a model $\Theta_{k_1,k_2}\in{\cal M}_{k_1,k_2}$, corresponds to probabilistic models where the elements of $X$ are independent and Poisson distributed with parameters $x_{i,j}\sim Poisson(P_{\alpha(i),\beta(j)})$, where $i\in V_{\alpha(i)}$, $j\in U_{\beta(j)}$ in the model  $\Theta_{k_1,k_2}$. The corresponding likelihood is denoted as $P(X\mid  \Theta_{k_1,k_2})$, the actual probability of $X$ is denoted as $P(X\mid A)$. 
The maximum likelihood model is found from the program that maximizes the expected log-likelihood: 
\begin{align*}
\Theta_{k_1,k_2}^*
&= \argomax_{\Theta_{k_1,k_2}\in{\cal M}_{k_1,k_2}}\sum_X P(X\mid A) \log P(X\mid \Theta_{k_1,k_2})\\
&=\argomax_{\Theta_{k_1,k_2}\in{\cal M}_{k_1,k_2}}\sum_X \left(P(X|A) \log\frac{P(X|\Theta_{k_1,k_2})}{P(X|A)}+P(X|A) \log P(X|A)\right)\\
&=\argomax_{\Theta_{k_1,k_2}\in{\cal M}_{k_1,k_2}}(-D(P_A\mid\mid P_{\Theta_{k_1,k_2}})-H(P_A))
\end{align*}
where $D$ is the Kullback-Leibler divergence between distributions, $H$ denotes entropy and $P_A$ and $ P_{\Theta_{k_1,k_2}}$ are the two families of Poisson distributions for the matrix elements of $X$. Since $H(P_A) $ is independent of $\Theta_{k_1,k_2}$, it does not affect the identification of the maximum likelihood model. Thus, the final program for finding the optimal model is
\begin{equation}
\label{matrixfinalprog}
\Theta_{k_1,k_2}^*= \argomin_{\Theta_{k_1,k_2}\in{\cal M}}D(P_A\mid\mid P_{\Theta_{k_1,k_2}}).
\end{equation}
The description length of a model $l(\Theta_{k_1,k_2}\in{\cal
  M}_{k_1,k_2})$ consists of the description length $l(V)+l(U)$ of the
two partitions and the description length of the block average matrix
$l(P(X))$. For the latter we need to know only the integers presenting
the block sums of $X$, since the denominator is known for a fixed
partition $(U,V)$. The code lengths of such integers are, for large
matrices, simply the logarithms of the integers. For $l(U)+l(V)$ we
use the same entropy based formula as in \eqref{poibmcodexi}.  As a
result we end up with the following expression for the description
length of the random multi-graph model $A$ using the modeling space
${\cal M}_{k_1,k_2}$:
\begin{align*}
l_{k_1,k_2}(A)&=D(P_A\mid\mid P_{\Theta_{k_1,k_2}^*})+l(V^*)+l(U^*)\\
&\quad+ \sum_{1\leq \alpha\leq k_1;1\leq \beta\leq k_2}\E(\log(e_{\alpha, \beta}+1\mid  P_{\Theta_{k_1,k_2}}^* ),
\end{align*}
where 
$$
e_{\alpha,\beta}=\sum_{i\in V_\alpha^*, j\in U_\beta^*}x_{i,j}. 
$$
The star superscript refers to parameters corresponding to the solution of the program \eqref{matrixfinalprog}. The expectation of logarithm  is not explicitly computable. However, we assume large matrices and blocks, and then Jensen's inequality provides a tight upper bound that can be used in practical computations. Thus, the final expression for the description length of $A$ is
\begin{equation}
\label{matrixcaseobjfun}
l_{k_1,k_2}(A)=D(P_A\mid\mid P_{\Theta_{k_1,k_2}^*})+l(V^*)+l(U^*)+ \sum_{1\leq \alpha\leq k_1;1\leq \beta\leq k_2}\log(a_{\alpha, \beta}+1 ),
\end{equation}
where
$$
a_{\alpha,\beta}=\sum_{i\in V_\alpha^*, j\in U_\beta^*}a_{i,j}. 
$$
The full two-part MDL would now be realized by finding the global minimum of this expression over various $(k_1,k_2)$. We return to this case in the algorithm section \ref{algosec}. Although a heuristic one, we believe that our method for matrices is both reasonable and easy to use and implement, see \cite{reittuweissbazso}. 
\subsection{Refined MDL and asymptotic model complexity}

Let us next consider Rissanen's {\em refined MDL} variant
(see \cite{mdlgrunwald}). The idea is to generate just one distribution on
$\Omega_n$, called the {\em normalized maximum likelihood
  distribution} $P_{nml}$. Then a graph $G\in\Omega_n$ has the description
length $-\log P_{nml}(G)$ which is at most as large as the one given
by the two-part code in \eqref{twopartcode}. The function $P(\cdot\mid
\hat{\Theta}_k(\cdot))$ maps graphs of size $n$ into $[0,1]$, and it
is not a probability distribution, because $\sum_{G\in\Omega} P(G|
\hat{\Theta}_k(G))> 1$. However, a related true probability
distribution can be defined as
\begin{equation}
\label{nmldef}
P_{nml}(\cdot)
=\frac{P(\cdot\mid \hat{\Theta}_k(\cdot))}{
  \sum_{G\in\Omega} P(G| \hat{\Theta}_k(G))}.
\end{equation}

The problem with this is that a computation of the normalization factor
in \eqref{nmldef} is far too involved: finding a maximum
likelihood parametrization for a single graph is a `macroscopic'
computational task by itself and it is not possible to solve such a
problem explicitly for all graphs. Therefore the two-part variant is
a more attractive choice in a practical context. However, the refined
MDL approach is useful as an idealized target object for justifying
various approximate implementations of the basic idea. It appears
that in an asymptotic sense the problem is solvable for large simple
graphs. The logarithm of the normalization factor in \eqref{nmldef}
is called the {\em parametric complexity} of the model space ${\cal
  M}_{n/k}$:
\begin{eqnarray}
\label{parcompldef}
COMP({\cal M}_{n/k})
:= \log \left( \sum_{G\in\Omega_n} P(G\mid \hat{\Theta}_k(G))\right). 
\end{eqnarray}
In a finite modelling space case like in ours, this can be considered
as a definition of model complexity. We have now the following simple
bounds:
\begin{proposition}
\label{nmlboundprop}
$$\log\left(S2(n,k)\right)\leq COMP({\cal M}_{n/k})\leq m_k+1,$$ where
  we use the same notation as in Proposition \ref{twopartprop}.
\end{proposition}
\begin{proof}
The lower bound follows from the fact that we can have at least this
number of graphs that have likelihood 1 in ${\cal M}_{n/k}$. This
corresponds to graphs for which the nodes can be partitioned into $k$
non-empty sets and inside each set we have a full graph and no links
between the distinct sets. Thus, for every partition there is at least one graph that
has likelihood one and all such graphs are different from each other
since there is a bijection between those graphs and partitions.

For the upper bound, we notice that according to Proposition
\ref{twopartprop}, there is a prefix coding with code lengths that
correspond to the two-part code. As a result, Kraft's inequality yields
that $\sum_{G\in\Omega_n}2^{-l_k(G)}\leq 1$, or
\begin{equation*}
1\geq \sum_{G\in\Omega_n}2^{-\lceil -\log P(G\mid\hat{\Theta}_k(G))\rceil-m_k}
\geq  \sum_{G\in\Omega_n}2^{\log P(G\mid\hat{\Theta}_k(G))-1-m_k},
\end{equation*}
from which we get
\begin{equation*}
\sum_{G\in\Omega_n} P(G\mid\hat{\Theta}_k(G))\leq 2^{m_k+1}.
\end{equation*}
Taking logarithms, we arrive at the claimed upper bound.
\end{proof}

When considering large-scale structures corresponding to moderate $k$,
the upper and lower bounds in Proposition \ref{nmlboundprop} are
asymptotically equivalent, and we have
\begin{corollary}
\label{asymodcompcor}
Assume that $k>1$ is fixed. Then
  \begin{equation*}
COMP({\cal M}_{n/k})\sim n \log k,\quad n\rightarrow\infty.    
  \end{equation*}
\end{corollary}
\begin{proof}
Denoting the lower and upper bound of parametric complexity in
Proposition \ref{nmlboundprop} respectively by $b_l$ and $b_u$, we
argue that $b_u\sim b_l\sim n\log k$ asymptotically when $n
\rightarrow\infty$. This follows from the fact that the dominant
asymptotic component of both $b_u$ and $b_l$ is $\log
S2(n,k)$. Indeed, $S2(n,k)\sim \frac{k^n}{k!}$ for fixed $k$, the
asymptotic of $\log S2(n,k) $ is linear in $n$, and all other terms of the
asymptotics of both bounds are additive and at most logarithmic in $n$.
\end{proof}

\begin{remark}
The speed of convergence of the upper and lower bounds in Proposition
\ref{nmlboundprop} is of type $\log n/n$.
\end{remark}

%\begin{remark}
%The leading asymptotic of the parametric complexity is $\sim n \log k$
%and the dependency on $k$ is just logarithmic. The other terms have
%quadratic dependency on $k$, however, their coefficient is sub-linear in
%$n$ and thus the large-scale optimum is indeed defined by this simple
%term.
%\end{remark}

\subsection{$\epsilon$-regularity {\em vs.}\ stochastic block models}\label{szemblomodifsec}

Although the structure that a MDL-based algorithm finds typically
looks like an $\epsilon$-regular structure, there is a principal
difference. In particular cases, an $\epsilon$-regular graph can have
a structure that allows much better compression than that provided by
the $\epsilon$-regular partition. In this section we give an explicit
example of such a case.

An important point in SRL is that for any $\epsilon>0$, there is an
upper bound for the size of regular partition, $M(\epsilon)$ so that
for any graph with size above some finite threshold $N(\epsilon)$, all
such graphs have a regular partition with at most $M(\epsilon)$
sets. Based on this, we show that the $\epsilon$-regular structure of
SRL and the structure induced by the MDL need not coincide. Let us fix
an order of graph $2n$, large enough so that SRL holds for some
$\epsilon>0$, and that $M(\epsilon)<n^{1-\alpha}$ for some
fixed $0 <\alpha< 1/2$.

\begin{proposition}
There is a graph of order $2n$ such that it has a MDL structure with
code length $o(n^2)$ and an $\epsilon-$regular structure that allows
only $\Theta(n^2)$ code length, where $\Theta(n^2)$ denotes any strictly linear function of $n^2$. 
\end{proposition}
\begin{proof}
Take $n$ large enough as prescribed above. Then construct a bipartite
graph $(X,Y)$ with $|X|=|Y|=n$ such that $n$ is divisible by
$n^\alpha$ with some rational $\alpha\in(0,1/2)$. Assume that both parts of the
bipartition are further partitioned into equal size blocks: $X=\sum
X_i, Y=\sum Y_i$, $|X_i|=|Y_i|=n^\alpha$. Define then a random graph
$G_p=(X,Y,E)$ as follows. For each pair $(X_i,Y_j)$, take
$e(X_i,Y_j)=\xi_{i,j}|X_i|Y_j|$, where $\xi_{i,j}\sim
Ber(p)$ is a Bernoulli random variable with parameter $0<p<1$, and the
variables for different pairs are independent. Assume that there are
no other edges.

We show that, with high probability, $G_p$ is $\epsilon$-regular with
regular partition $\set{X,Y}$. In the $MDL$-approach, such a structure
has a coding length at least $\binom{n}{2}H(p+o(1))=O(n^2)$. This
comes from the log-likelihood part, and the $o(1)$ corresponds to very
small deviations of link densities from the expected value $p$ that
can be made arbitrarily small by increasing $n$. Now we check that the
$\epsilon$-regularity of graph $G_p$ has a positive probability, which
implies that such an $\epsilon$-regular pair exists (actually, it
appears that this happens with high
probability). $\epsilon$-regularity means that for any $X'\subseteq
X$, $Y'\subseteq Y$, $\mid X'\mid,$ $\mid Y'\mid >\epsilon n$, the
link density $d(X',Y')$ deviates from the link density of the pair,
$d(X,Y)$, no more than by $\epsilon$. By definition,
$$
d(X',Y')=\frac{\sum_{i,j}\xi_{i,j}\mid X'\cap X_i\mid \mid Y'\cap Y_j\mid}{\mid X'\mid\mid Y'\mid},
$$
and as a result the expectation is
$$
\E d(X',Y')=\frac{\sum_{i,j}\E \xi_{i,j}\mid X'\cap X_i\mid \mid Y'\cap Y_j\mid}{\mid X'\mid\mid Y'\mid}=p.
$$
Denote
$$
x_{i,j}:=\frac{\xi_{i,j}\mid X'\cap X_i\mid \mid Y'\cap Y_j\mid}{\mid X'\mid\mid Y'\mid}
$$ The range of $x_{i,j}$ is interval $[0,1]$ of unit
length. Hoeffding's inequality yields for $S=\sum x_{i,j}$ that
$$
P(\mid S-\E S\mid>t)\leq 2e^{-\frac{t^2}{\sum b_{i,j}^2}},
$$
where $b_{i,j}=\frac{\mid X'\cap X_i\mid \mid Y'\cap Y_j\mid}{\mid X'\mid\mid Y'\mid}$ is the range of variable $x_{i,j}$. The denominator of the exponent in the right hand-side of the Hoeffding inequality can be bounded as
$$
\sum b_{i,j}^2\leq\frac{1}{(\epsilon n)^2} 
  \sum \mid X'\cap X_i\mid^2 \mid Y'\cap Y_j\mid^2
\leq\frac{1}{(\epsilon n )^2} n^{2\alpha}n^{2\alpha}\left(\frac{n}{n^\alpha}\right)^2= \frac{1}{\epsilon^2}n^{2\alpha}.
$$
 By taking $t=\epsilon \mid X'\mid \mid Y'\mid\geq \epsilon (\epsilon n)^2$ we get for the link density:
$$
P(\mid d(X',Y')-p\mid>\epsilon)\leq 2e^{-2\epsilon^8 n^{2(1-\alpha  )}}.
$$
Finally, since there are at most $4^n$ pairs of subsets $(X',Y')$ from which to choose, the probability that none of them violates regularity is lower bounded by 
$$
1-2\times 4^ne^{-2\epsilon^8 n^{2(1-\alpha  )}}\rightarrow 1,
$$ if the exponent has a positive power of $n$, and this happens when
$\alpha<1/2$. Thus, all large subsets have densities that deviate from
expectation less than $\epsilon$ with a probability tending to
one. Thus, we have shown the $\epsilon$-regularity of the partition $(X,Y)$. 

On the other hand, using MDL, we could reach the level of small sets
$X_i$ and $Y_i$, and the corresponding log-likelihood is zero. The
model complexity is $o(n^2)$, as can be easily seen from asymptotic
formulas for the upper bound for $m_k$, with $k=n^{1-\alpha}$.
\end{proof}

\section{The Regular Decomposition approach to stochastic block models}\label{resultsec}

\subsection{Block model codes}
\label{blomocodesec}

The previous section developed both the two-part and refined variants
of the MDL theory, as presented in \cite{mdlgrunwald}, for the model
space of stochastic block models. In the following, we formulate a
variant of two-part MDL that allows both practical implementations and
a proof of consistency, i.e., that the MDL principle identifies a 
correct block model. It was shown above that the most difficult task in the
description of a block model is identifying the partition. The same is
true for the model complexity, which is asymptotically just the
logarithm of the number of partitions. It appears that in order to
prove consistency, we need quite a delicate estimate for the
description length of the partition. The asymptotic model complexity
given in Corollary \ref{asymodcompcor} seems to be too crude for proof of 
consistency. A full resolution of this intriguing question is left for
further investigations.

We call our two-part MDL construction a {\em block model code} of a
graph with respect to a partition of its nodes that allows the
computation of a tight upper bound of the code length. This upper
bound is also consistent with a more generic information theoretic
point of view with a semi-constructive coding scheme.

We denote by $H(\cdot)$ both Shannon's entropy function of a partition
and the entropy of a Bernoulli distribution, i.e.
$$
H(\xi)=-\sum_{A\in\xi}\frac{|A|}{|V|}\log\frac{|A|}{|V|},\quad 
H(p)=-p\log p-(1-p)\log(1-p).
$$ 

\begin{remark}\label{natlogrem}
\rm In the rest of this paper, we define also information-theoretic
functions in terms of natural logarithms, and certain notions like
code lengths should be divided by $\log 2$ to obtain their values in
bits.
\end{remark}

\begin{definition}
\label{bmcodedef}
A {\em block model code} of a graph $G=(V,E)$ with respect to a
partition $\xi$ of $V$ is a code with the following structure:\\

{\em The model part:}
\begin{itemize}
\item first, the sizes of the blocks $A\in\xi$ are given as integers;
\item second, the edge density $d(A)$ inside each block $A\in\xi$ and
  the edge density $d(A,B)$ between each pair of distinct blocks
  $A,B\in\xi$ are given as the numerators of the rational numbers
  presenting the exact densities.
\end{itemize}
The aim of these two codes is to describe the parameters of two
probability distributions, one for the links and the other for the membership
of nodes in the blocks of the partition.\\

{\em The data part:}
\begin{itemize}
\item third, the partition $\xi$ is specified by a prefix code corresponding to membership distribution $P(i\in A)=|A|/n$, where all nodes are independent of each other;
\item fourth, the edges inside each block $A\in\xi$ are specified by a
  prefix code corresponding to a stochastic block model distribution of
  links inside each block of $\xi$;
\item fifth, the edges between each pair of blocks $A,B\in\xi$ are
  specified by a prefix code corresponding to a block model distribution
  of links between pairs of blocks in $\xi$.
\end{itemize}
\end{definition}

The description of link densities as link probabilities (the second code) is
natural, since conditionally to a partition the stochastic block model
is just a collection of Bernoulli models, where the best choice is to use
averages as parameters. Note that a block model code can be given for
any graph with respect to any partition of its nodes.

From Kraft's inequality and the above definitions, it follows that
there exists a prefix code for a graph $G=(V,E)$ with respect to a
partition $\xi=\set{A_1,\ldots,A_k}$ of $V$ with length at most
(and, for large graphs, typically close to)
\begin{align}
\label{bmcodexi}
\nonumber
L(G|\xi)&=L_1(G|\xi)+L_2(G|\xi)+L_3(G|\xi)+L_4(G|\xi)+L_5(G|\xi),\\
\nonumber
L_1(G|\xi)&=\sum_{i=1}^kl^*(|A_i|),\\
L_2(G|\xi)&=\sum_{i=1}^kl^*\left(\binom{|A_i|}{2}d(A_i)\right)
+\sum_{i<j}l^*\left(|A_i||A_j|d(A_i,A_j)\right),\\
\nonumber
L_3(G|\xi)&=|V|H(\xi),\\
\nonumber
L_4(G|\xi)&=\sum_{i=1}^k\binom{|A_i|}{2}H(d(A_i)),\\
\nonumber
L_5(G|\xi)&=\sum_{i<j}|A_i||A_j|H(d(A_i,A_j)),
\end{align}
where $l^*(m)$ was defined by \eqref{rissanenl}.  Below we shall
approximate $l^*(m)$ by $\log m$ without further mentioning, because
their difference is insignificant in our context. Similarly, we have
dropped ceiling functions systematically. Also, recall Remark
\ref{natlogrem} on the use of natural logarithms.

Next, we shall define block model codes for the Poissonian block models of Definition \ref{poiblomodef}. The entries of the random matrix $E$ are Poisson distributed integers. For a pair of disjoint sets $A,B\subset V$, the set $\eset{e_{ij}}{i\in A,\,j\in B}$ is a sample from a distribution $R=(r_\ell)_{\ell\ge0}$ that is mixture of Poisson distributions. It would be hard to encode the sample by first estimating the unknown mixture distribution. Instead, we base the code simply on the sample mean \begin{equation*}
e_{AB}=\frac{1}{|A||B|}\sum_{i\in A,\ j\in B}e_{ij}
\end{equation*} 
and encode $\eset{e_{ij}}{i\in A,\,j\in B}$ as if it came from a Poisson distribution $P=(p_\ell)$ with parameter $e_{AB}$. Thus, by Kraft's inequality, a value $e_{ij}=\ell$ can be well encoded by a codeword with approximate length
\begin{equation*}
-\log p_\ell
=-\log\left(\frac{e_{AB}^\ell}{\ell!}e^{-e_{AB}}\right).
\end{equation*}
By the fundamental information inequality
\begin{equation*}
\sum_\ell r_\ell(-\log p_\ell)
\ge\sum_\ell r_\ell(-\log r_\ell)=H(R),
\end{equation*}
this encoding is suboptimal for arbitrary disjoint subsets $A$ and $B$, but it is optimal when $A$ and $B$ are blocks of the model partition $\xi$ and $R$ presents a pure Poisson distribution. Thus, the suboptimality only improves the contrast between $\xi$ and other partitions of $V$.
%Denote by $H_P(a)$ the entropy of the distribution
%$\mbox{\it Poisson}(a)$:
%\begin{equation*}
%H_p(a)=-\sum_{k=0}^\infty\frac{a^k}{k!}e^{-k}
%  (k\log a-\log k!-a).
%\end{equation*}

For an arbitrary partition $\eta=\set{B_1,\ldots,B_m}$, the encoding of all $e_{ij}'s$ with the above rule requires about \begin{align*}
&\sum_{i=1}^m\binom{|B_i|}{2}\sum_{\ell\ge0}
r_\ell^{(B_i)}\left(-\log\left(\frac{e_{B_i}^\ell}{\ell!}e^{-e_{B_i}}\right)\right)\\
&+\sum_{i<j}|B_i||B_j|\sum_{\ell\ge0}
r_\ell^{(B_iB_j)}\left(-\log\left(\frac{e_{B_iB_j}^\ell}{\ell!}e^{-e_{B_iB_j}}\right)\right)\\
=&\sum_{i=1}^m\binom{|B_i|}{2}(-\log e_{B_i})\sum_{\ell\ge0}\ell r_\ell^{(B_i)}
+\sum_{i<j}|B_i||B_j|(-\log e_{B_iB_j})\sum_{\ell\ge0}\ell
r_\ell^{(B_iB_j)}\\
&+\sum_{i=1}^m\binom{|B_i|}{2}\sum_{\ell\ge0}
r_\ell^{(B_i)}\log \ell!
+\sum_{i<j}|B_i||B_j|\sum_{\ell\ge0}
r_\ell^{(B_iB_j)}\log\ell!\\
&+\sum_{i=1}^m\binom{|B_i|}{2}e_{B_i}
+\sum_{i<j}|B_i||B_j|e_{B_iB_j}\\
=&\sum_{i=1}^m\binom{|B_i|}{2}\phi(e_{B_i})
+\sum_{i<j}|B_i||B_j|\phi(e_{B_iB_j})\\
&+\frac{1}{2}\sum_{v,w\in V,\,v\not=w}\log e_{vw}!
+\sum_{v,w\in V,\,v\not=w}e_{vw},
\end{align*}
where $r$-symbols refer to the empirical distribution of the $e_{vw}$s,
\begin{equation*}
e_{B_i}:=\frac{\sum_{v,w\in B_i}e_{vw}}{|B_i|(|B_i|-1)},\quad
e_{B_iB_j}:=\frac{\sum_{v,w\in B_i}e_{vw}}{|B_i||B_j|}.
\end{equation*}
and 
\begin{equation*}
\phi(x)=-x\log x.    
\end{equation*}
Now, we define the {\em Poissonian block model code
  length} of $E$ with respect to any partition
$\eta=\set{B_1,\ldots,B_m}$ as
\begin{align}
\label{poibmcodexi}
\nonumber
L(E|\eta)&=L_0(E|\eta)+L_1(E|\eta)+L_2(E|\eta)+L_3(E|\eta)+L_4(E|\eta)+L_5(E|\eta),\\
\nonumber
L_0(E|\eta)&=\frac{1}{2}\sum_{v,w\in V,\,v\not=w}\log e_{vw}!
+\sum_{v,w\in V,\,v\not=w}e_{vw},\\
\nonumber
L_1(E|\eta)&=\sum_{i=1}^m\log(|B_i|),\\
L_2(E|\eta)&=\sum_{i=1}^m\log\left(\binom{|B_i|}{2}e_{B_i}\right)
+\sum_{i<j}\log\left(|B_i||B_j|e_{B_iB_j}\right),\\
\nonumber
L_3(E|\eta)&=|V|H(\eta),\\
\nonumber
L_4(E|\eta)&=\sum_{i=1}^m\binom{|B_i|}{2}\phi(e_{B_i}),\\
\nonumber
L_5(E|\eta)&=\sum_{i<j}|B_i||B_j|\phi(e_{B_iB_j}).
\end{align}
Note that the term $L_0(E|\eta)$ is independent of the partition $\eta$ and can be neglected when minimizing over $\eta$. 

%\begin{remark}
%To put it more simply, any matrix, $A$, with non-negative %elements can be interpreted as expectation of a random %Poisson matrix: $EX=A$, where elements of $X$ are %independent and Poisson distributed random variables with %$E(X)_{i,j}=a_{i,j}$. 
%\end{remark}

The MDL for matrices, derived in Section \ref{twoparttheosecmatrix}, is essentially equivalent with \eqref{poibmcodexi}. Indeed, the first term of \eqref{matrixcaseobjfun} can be written as
\begin{equation*}
D(P_A\mid\mid P_{\Theta_{k_1,k_2}^*})
=\sum_{\alpha=1}^{k_1}\sum_{\beta=1}^{k_2}|V_\alpha||U_\beta|
  \phi(\bar{a}_{\alpha,\beta})+\sum_{i,j}a_{i,j}\log a_{i,j},
\end{equation*}
where $a_{i,j}$ and $a_{\alpha,\beta}$ are as in \eqref{matrixcaseobjfun} and $\bar{a}_{\alpha,\beta}=a_{\alpha,\beta}/(|V_\alpha||U_\beta|)$.
The latter sum does not depend on the partitions and can be neglected in minimization. 

\subsection{Accuracy of block structure identification by MDL}\label{theoremsec}

The general idea of the Regular Decomposition method is to be a
generic tool for separating structure and randomness in large data
sets of graph or matrix form. A partition of a real-world data set
that minimizes the (nominal) code length given in \eqref{bmcodexi}
resp.\ \eqref{poibmcodexi} can often not be compared with a `true
solution' for the simple reason that there may not be any objective
notion of a `true structure' of the data. However, it is important to
analyse and understand how the method performs when the data really
originates from a stochastic block model. This question is
called the consistency of MDL. Our results on this question are summarized
in the following theorem, formulated in terms of the asymptotic
behavior of a model sequence as specified by Construction
\ref{graphseqconstr}. In such a framework, an event is said to happen
    {\em with high probability}, if its probability tends to 1 when
    $n\to\infty$.

\begin{thm}
\label{blomothm}
Consider a sequence of stochastic block models $(G_n,\xi_n)$ based on
a vector $(\gamma_1,\ldots,\gamma_k)$ of relative block sizes and a
matrix $D=(d_{ij})_{i,j=1}^k$ of link probabilities, as described in
Definition \ref{blomodef} and Construction \ref{graphseqconstr}. With
high probability, the following hold:
\begin{enumerate}
\item\label{blomothmkclaim}Among all partitions $\eta$ of $V_n$ such
  that $|\eta|\le k$, $\xi_n$ is the single minimizer of
  $L(G_n|\eta)$.
\item\label{blomothmbigpartclaim}For any fixed
  $\epsilon\in(0,\min_i\gamma_i)$, $\xi_n$ is the single minimizer of
  $L(G_n|\eta)$ among partitions $\eta$ with minimal block size larger
  than $n\epsilon$.
\item\label{blomothmrefinclaim}No refinement $\eta$ of $\xi_n$ with
  $|\eta|\le m$ improves $L(G_n|\xi_n)$ by more than
  $\mathit{const}(k,m)\log n$.
\end{enumerate}
The corresponding claims hold for the Poissonian block model {\em
  mutatis mutandis}.
\end{thm}
\begin{proof}
The proof is given in the next section and structured into several
propositions. Claim \ref{blomothmkclaim} follows from Proposition
\ref{neartoendprop}, Proposition \ref{bigdifprop} and Corollary
\ref{concacor}. Claim \ref{blomothmbigpartclaim} is Proposition
\ref{bigblockprop}. Claim \ref{blomothmrefinclaim} is Proposition
\ref{refinementprop}.
\end{proof}

The results of Section \ref{blomothmsec} offer a richer picture than
what was distilled into Theorem \ref{blomothm}. For example,
Proposition \ref{neartoendprop} shows that if $|\eta|=k$ and $\eta$
differs from $\xi_n$ only a little, then the remaining misplaced nodes
can be immediately identified by computing their effect to the value
of $L(G_n|\eta)$.  On the other hand, we have not been able to exclude
the possibility that a refinement of $\eta$ could yield a slight
$O(\log n)$ improvement of the code length.

\begin{remark}
It is rather obvious that with large $n$, the identification of the
block structure $\xi_n$ is robust against independent noise. The
simplest case is that the Poissonian block model is disturbed by
additive Poissonian noise to each matrix element:
\begin{equation*}
  \tilde{e}_{ij}=e_{ij}+\phi_{ij},
\end{equation*}
where the $\phi_{ij}$s are i.i.d.\ with
$\phi_{ij}\sim\mathit{Poisson}(\nu)$, $\nu>0$. Then $(\tilde{e}_{ij})$
is again an irreducible Poissonian block model with the same
partition. More interesting cases are binary flips in the graph case and
multiplicative noise with mean 1 in the case of non-negative
matrices. We leave these for forthcoming work.
\end{remark}

\section{Proof of Theorem \ref{blomothm}}\label{blomothmsec}

Through this section, we consider a sequence $(G_n|\xi_n)$ of
increasing versions of a fixed stochastic block model based on a
vector $(\gamma_1,\ldots,\gamma_k)$ of relative block sizes and a
matrix $D=(d_{ij})_{i,j=1}^k$ of link probabilities, as specified in
Construction \ref{graphseqconstr}.

Consider partitions $\eta$ of $V_n$. 
%such that $|\eta|=k$, $\eta=\set{B_1,\ldots,B_k}$. 
Denote
%\begin{equation*}
%\mathfrak{d}(\eta,\xi_n)=\min_{\pi\in\Pi_k}
%\frac{1}{n}\sum_{i=1}^k\max\set{|A_i\setminus B_{\pi(i)}|,|B_{\pi(i)}\setminus A_i|},
%\end{equation*}
\begin{equation*}
\mathfrak{d}(\eta,\xi_n)=
\frac{1}{n}\max_{B\in\eta}\min_{A\in\xi_n}|B\setminus A|.
\end{equation*}
%where $\Pi_k$ is the set of permutations of $\set{1,\ldots,k}$. Thus,
%$\mathfrak{d}(\eta,\xi_n)$ expresses how much the two partitions
%differ in terms of their best block-wise pairing. In the following, we
%simplify notation by assuming that the blocks of $\eta$ are numbered
%according to such an optimal pairing. 
Thus, $\mathfrak{d}(\eta,\xi_n)=0$ if and only if $\eta$ is a
refinement of $\xi_n$. If $v\in V_n$ and $B\in\eta$,
denote by $\eta_{v,B}$ the partition obtained from $\eta$ by moving
node $v$ to block $B$ (if $v\in B$, then $\eta_{v,B}=\eta$).

\begin{proposition}
\label{neartoendprop}
There is a number $\epsilon_0>0$ such that the following holds with high
probability: if $|\eta|=k$ and $\mathfrak{d}(\eta,\xi_n)\le\epsilon_0$, then 
\begin{equation*}
\mbox{if }A\in\xi_n,\ B\in\eta,\ 
\frac{1}{n}|B\setminus A|\le\epsilon_0\mbox{ and }
v\in A\setminus B,\mbox{ then }L(G_n|\eta_{v,B})<L(G_n|\eta).
\end{equation*}
\end{proposition}
\begin{proof}
Let $\epsilon,\delta>0$ be small numbers and $m$ a positive integer
to be specified. They can be chosen so that the following holds:
\begin{itemize}
\item $\epsilon$ is small so that $\eta$ and $\xi_n$ nearly overlap when
  $\mathfrak{d}(\eta,\xi_n)\le\epsilon$:
  \begin{equation}
\label{nearityineq}
  mn\epsilon\le\delta\min_{A\in\xi_n}|A|;  
  \end{equation}
\item all the differing link probabilities are widely separated in
  $\delta$ units:
  \begin{equation}
\label{neardensdifineq}
  \delta\le\frac{1}{m}
\min\eset{|d_{ij_1}-d_{ij_2}|}{i,j_1,j_2\in\set{1,\ldots,k},\ d_{ij_1}\not=d_{ij_2}};   \end{equation}
  \item the empirical densities are close to their mean values: for
    any $A_i,A_j\in\xi_n$ (possibly $i=j$), we have, with high
    probability,
  \begin{equation}
\label{nearempdensineq}
    |d(A_i,A_j)-d_{ij}|+m\epsilon\le\delta.
  \end{equation}
\end{itemize}
Let $\eta$ be a partition of $V_n$ such that
$\mathfrak{d}(\eta,\xi_n)\le\epsilon$. Condition \eqref{nearityineq}
entails that for each block $A_i\in\xi_n$ there is a unique block
$B_i\in\eta$ such that $|B_i\setminus A_i|\le\epsilon$. Let us now assume
that $v\in A_i\cap B_j$ and $i\not=j$, and compare the partitions
$\eta$ and $\eta_{v,B_i}$. Denote $b_i=|B_i|$, $i=1,\ldots,k$, and $\tilde{B}_i=B_i\cup\set{v}$, $\tilde{B}_j=B_j\setminus\set{v}$. Then
%\begin{align*}
\begin{equation}
\label{near45split}
\begin{split}
&L_4(G_n|\eta)+L_5(G_n|\eta)-(L_4(G_n|\eta_{v,B_i})+L_5(G_n|\eta_{v,B_i}))\\
&=\binom{b_i}{2}H(d(B_i))+\binom{b_j}{2}H(d(B_j))
  -\binom{b_i+1}{2}H(d(\tilde{B}_i))-\binom{b_j-1}{2}H(d(\tilde{B}_j))\\
&\quad+\sum_{q\not=i,j}\Big[b_ib_qH(d(B_i,B_q))+b_jb_qH(d(B_j,B_q))\\
&\quad\quad-(b_i+1)b_qH(d(\tilde{B}_i,B_q))
  +(b_j-1)b_qH(d(\tilde{B}_j,B_q))\Big]\\
&\quad+b_ib_jH(d(B_i,B_j))-(b_i+1)(b_j-1)H(d(\tilde{B}_i,\tilde{B}_j)).
\end{split}
\end{equation}
%\end{align*}
Consider first the sum over $q$. Leaving out the common factor $b_q$, each term of the sum can be written as
\begin{align*}
&\quad b_j\big[H(d(B_j,B_q))-\frac{b_j-1}{b_j}H(d(\tilde{B}_j,B_q))\big]\\
&\quad-(b_i+1)\big[H(d(\tilde{B}_i,B_q))-\frac{b_i}{b_j+1}H(d(B_i,B_q))\big]\\
&=b_j\bigg[H\left(\frac{b_j-1}{b_j}d(\tilde{B}_j,B_q)
  +\frac{1}{b_j}d(\set{v},B_q)\right)\\
&\quad\quad-\frac{b_j-1}{b_j}H(d(\tilde{B}_j,B_q))
  -\frac{1}{b_j}H(d(\set{v},B_q))\bigg]\\
&-(b_i+1)\bigg[H\left(\frac{b_i}{b_i+1}d(B_i,B_q)
  +\frac{1}{b_i+1}d(\set{v},B_q)\right)\\
&\quad\quad-\frac{b_i}{b_i+1}H(d(B_i,B_q))
  -\frac{1}{b_i+1}H(d(\set{v},B_q))\bigg]
\end{align*}
(note the addition and subtraction of the term
$H(d(\set{v},B_q))$). Using Lemmas \ref{infodifflem} and \ref{binmaxlemma}, and the assumptions
on $\epsilon$, $\delta$ and $m$,
%{\color{red}[WORK OUT DETAILS IF NEEDED!]} 
the last expression can be set to be, with high probability,
arbitrarily close to the number
\begin{equation*}
D_B(d_{iq}\|d_{jq})-D_B(d_{iq}\|d_{iq})=D_B(d_{iq}\|d_{jq})
\end{equation*}
(the function $D_B(\cdot\|\cdot)$, the Kullback-Leibler divergence of
Bernoulli distributions, is defined in \eqref{Ixpdef}).  Thus, the
  sum over $q$ is, with high probability, close to
\begin{equation*}
  b_q\sum_{q\not=i,j}D_B(d_{iq}\|d_{jq}).
\end{equation*}

Let us then turn to the remaining parts of \eqref{near45split} that
refer to two codings of the internal links of $B_i\cup B_j$. Similarly
as above, we can add and subtract terms to transform these parts into
\begin{align*}
&\binom{b_j}{2}\Bigg[H\left(\frac{\binom{b_j-1}{2}}{\binom{b_j}{2}}
 d(\tilde{B}_j)+\frac{b_j-1}{\binom{b_j}{2}}d(\set{v},\tilde{B}_j)\right)\\
&\quad\quad-\frac{\binom{b_j-1}{2}}{\binom{b_j}{2}}H(d(\tilde{B}_j))
  -\frac{b_j-1}{\binom{b_j}{2}}H(d(\set{v},\tilde{B}_j))\Bigg]\\
&-\binom{b_i+1}{2}\Bigg[H\left(\frac{\binom{b_i}{2}}{\binom{b_i+1}{2}}
 d(B_i)+\frac{b_i}{\binom{b_i+1}{2}}d(\set{v},B_i)\right)\\
&\quad\quad-\frac{\binom{b_i}{2}}{\binom{b_i+1}{2}}H(d(B_i))
  -\frac{b_i}{\binom{b_i+1}{2}}H(d(\set{v},B_i))\Bigg]\\
&+b_ib_j\Bigg[H\left(\frac{b_j-1}{b_j}d(B_i,\tilde{B}_j)
  +\frac{1}{b_j}d(\set{v},B_i)\right)\\
&\quad\quad-\frac{b_j-1}{b_j}H(d(B_i,\tilde{B}_j))
  -\frac{1}{b_j}H(d(\set{v},B_i))\Bigg]\\
&-(b_i+1)(b_j-1)\Bigg[H\left(\frac{b_i}{b_i+1}d(B_i,\tilde{B}_j)
  +\frac{1}{b_i+1}d(\set{v},\tilde{B}_j)\right)\\
&\quad\quad-\frac{b_i}{b_i+1}H(d(B_i,\tilde{B}_j))
  -\frac{1}{b_i+1}H(d(\set{v},\tilde{B}_j))\Bigg]\\
&\approx(b_j-1)D_B(d_{ij}\|d_{jj})-(b_i+1)D_B(d_{ii}\|d_{ii})
  +b_iD_B(d_{ii}\|d_{ij})-(b_j-1)D_B(d_{ij}\|d_{ij})\\
&\approx b_jD_B(d_{ij}\|d_{jj})+b_iD_B(d_{ii}\|d_{ij}).
\end{align*}
By the above analysis of \eqref{near45split}, we have obtained
\begin{align}
\nonumber
&L_4(G_n|\eta)+L_5(G_n|\eta)-(L_4(G_n|\eta_{v,B_i})+L_5(G_n|\eta_{v,B_i}))\\
\label{nearkappasum}
&\approx b_q\sum_{q\not=i,j}D_B(d_{jq}\|d_{iq})
+b_jD_B(d_{ij}\|d_{jj})+b_iD_B(d_{ii}\|d_{ij}).
\end{align}
By the irreducibility assumption \eqref{irreducibility}, there is a
block $A_q$ such that $d_{qi}\not=d_{qj}$, with the possibility that
$q\in\set{i,j}$. It follows that at least one of the $D_B(x\|y)$'s
in \eqref{nearkappasum} is positive. Denote
\begin{equation*}
\kappa^*=\min\eset{D_B(d_{ij_1}\|d_{ij_2})}{d_{ij_1}\not=d_{ij_2}}.
\end{equation*}
Thus, with high probability,
\begin{equation*}
L_4(G_n|\eta)+L_5(G_n|\eta)-(L_4(G_n|\eta_{v,B_i})+L_5(G_n|\eta_{v,B_i}))
>\frac{1}{2}(\kappa^*\min_i\gamma_i)n.
\end{equation*}
On the other hand, it is easy to compute that
\begin{equation*}
L_3(G_n|\eta)-L_3(G_n|\eta_{v,B_i})
=n(H(\eta)-H(\eta_{v,B_i}))\to\log\frac{\gamma_j}{\gamma_i}.
\end{equation*}
The changes of $L_1$ and $L_2$ when moving from $\eta$ to
$\eta_{v,B_i}$ are negligible. This concludes the proof.
\end{proof}

The proof of Proposition \ref{neartoendprop} showed that when
$\mathfrak{d}(\eta,\xi_n)\le\epsilon_0$, moving any node to its
correct block decreases $L(G_n|\eta)$ at least by
$(1/2)(\kappa^*\min_i\gamma_i)n$. In particular, with high probability,
$\xi_n$ is the unique minimizer of $L(G_n|\eta)$ among $k$-partitions
$\eta$ satisfying $\mathfrak{d}(\eta,\xi_n)\le\epsilon_0$.

\begin{proposition}
\label{bigdifprop}
For any $\epsilon\in(0,1)$ and positive integer $m$, there is a
constant $c_{\epsilon}>0$ such that the following holds with high
probability:
\begin{equation*}
\mbox{if }|\eta|\le m\mbox{ and }\mathfrak{d}(\eta,\xi_n)>\epsilon,\mbox{ then }
\frac{1}{n^2}(L(G_n|\eta)-L(G_n|\xi_n\vee\eta))\ge c_{\epsilon}.
\end{equation*}
\end{proposition}
\begin{proof}
Fix an $\epsilon\in(0,1)$ and let $\eta$ be a partition of $V_n$ such
that $\mathfrak{d}(\eta,\xi_n)>\epsilon$. By the concavity of $H$, we have
\begin{equation}
\label{bigdifsplit}
\begin{split}
&L_4(G_n|\eta)+L_5(G_n|\eta)\\
&=\sum_{B\in\eta}\binom{|B|}{2}H(d(B))+\frac12\sum_{
\begin{array}{c}
\scriptstyle      B,B'\in\eta\\
\scriptstyle      B\not=B'
    \end{array}
}|B||B'|H(d(B,B'))\\
&\ge\sum_{B\in\eta}\sum_{A\in\xi_n}\binom{|A\cap B|}{2}H(d(A\cap B))\\
&\quad+\sum_{B\in\eta}\frac12\sum_{
\begin{array}{c}
\scriptstyle      A,A'\in\xi_n\\
\scriptstyle      A\not=A'
    \end{array}
}|A\cap B||A'\cap B|H(d(A\cap B,A'\cap B))\\
&\quad+\frac12\sum_{
\begin{array}{c}
\scriptstyle      B,B'\in\eta\\
\scriptstyle      B\not=B'
    \end{array}
}\sum_{A,A'\in\xi_n}|A\cap B||A'\cap B'|H(d(A\cap B,A'\cap B'))\\
&=L_4(G_n|\eta\vee\xi_n)+L_5(G_n|\eta\vee\xi_n).
\end{split}
\end{equation}
By assumption, there is $B\in\eta$ such that $|B\setminus A|>\epsilon
n$ for every $A\in\xi_n$. It is easy to see that there must be (at
least) two distinct blocks, say $A_i$ and $A_j$, such that
\begin{equation}
\label{bigdifpiecesbig}
\min\set{|A_i\cap B|,|A_j\cap B|}\ge\frac{\epsilon}{k-1}n.
\end{equation}
By the irreducibility assumption \eqref{irreducibility}, there is a
block $A_q$ such that $d_{qi}\not=d_{qj}$, with the possibility that
$q\in\set{i,j}$. Fix an arbitrary $\delta>0$ to be specified later. By
$\epsilon$-regularity (claim \ref{regularityclaim} of Lemma
\ref{binmaxlemma}), with high probability, {\em every} choice of a
partition $\eta$ with $|B\setminus A|>\epsilon n$ results in some
blocks $A_i$, $A_j$, $A_q$ with the above characteristics plus the
regularity properties
\begin{equation}
\label{bigdifdensclose}
|d(A_i\cap B,A_q\cap B')-d_{iq}|\le\delta,\quad
|d(A_j\cap B,A_q\cap B')-d_{jq}|\le\delta,
\end{equation}
where $B'$ denotes a block of $\eta$ that maximizes $|A_q\cap B'|$
(note that because $|\eta|\le m$, $|A_q\cap B'|\ge|A_q|/m$). By the
concavity of $H$,
\begin{align*}
&|A_i\cap B||A_q\cap B'|H(d(A_i\cap B,A_q\cap B'))\\
&\quad+|A_j\cap B||A_q\cap B'|H(d(A_j\cap B,A_q\cap B'))\\
&=|(A_i\cup A_j)\cap B||A_q\cap B'|
  \bigg[\frac{|A_i\cap B|}{|(A_i\cup A_j)\cap B|}H(d(A_i\cap B,A_q\cap B'))\\
&\quad+\frac{|A_j\cap B|}{|(A_i\cup A_j)\cap B|}
  H(d(A_j\cap B,A_q\cap B'))\bigg]\\
&<|(A_i\cup A_j)\cap B||A_q\cap B'|H(d((A_i\cup A_j)\cap B,A_q\cap B')).
\end{align*}
In the case that $q\in\set{i,j}$ and $B=B'$, we obtain a similar
equation where $|A_q\cap B|$ is partly replaced by $|A_q\cap
B|-1$. Because of \eqref{bigdifdensclose} and \eqref{bigdifpiecesbig},
the difference between the sides of the equality has a positive lower
bound that holds with high probability. On the other hand, this
difference is part of the overall concavity inequality
\eqref{bigdifsplit}. 
%{\color{red}[ADD THE $\epsilon,\delta,m$ DETAILS]}
\end{proof}

\begin{proposition}
\label{genrefinementprop}
For any refinement $\eta$ of $\xi_n$, we have
\begin{equation}
\label{genrefinementgainub}
\begin{split}
&L_4(G_n|\xi_n)+L_5(G_n|\xi_n)-(L_4(G_n|\eta)+L_5(G_n|\eta))
\buildrel{(st)}\over\le\sum_{j=1}^{M(|\eta|)-M(k)}(\log 2+Y_i),
\end{split}
\end{equation}
where $(st)$ refers to stochastic order, the $Y_i$'s are i.i.d.\ {\it
  Exp}$(1)$ random variables, and
\begin{equation}
\label{Mdef}
M(x)=\frac{x(x+1)}{2}.
\end{equation}
\end{proposition}
\begin{proof}
Here we apply results presented in Appendix \ref{chernoffsec}. Denote
by $\eta\cap A_i$ the subset of $\eta$ whose members are subsets of
the block $A_i$ of $\xi_n$. Writing the edge code lengths of the
coarser and finer partition similarly as in \eqref{bigdifsplit},
taking the difference and using \eqref{HI}, we obtain
\begin{equation}
\label{genrefinementgain}
\begin{split}
&L_4(G_n|\xi_n)+L_5(G_n|\xi_n)-(L_4(G_n|\eta)+L_5(G_n|\eta))\\
&=\sum_{i=1}^k\Bigg(\sum_{B\in\eta\cap A_i}\binom{|B|}{2}D_B(d(B)\|d_{ii})
  +\frac12\sum_{
\begin{array}{c}
\scriptstyle      B,B'\in\eta\cap A_i\\
\scriptstyle      B\not=B'
    \end{array}
}|B||B'|D_B(d(B,B')\|d_{ii})\\
&\quad-\binom{|A_i|}{2}D_B(d(A_i)\|d_{ii})\Bigg)\\
&+\sum_{i<j}\Bigg(\sum_{B\in\eta\cap A_i}\sum_{B'\in\eta\cap A_j}|B||B'|
  D_B(d(B,B')\|d_{ij})-|A_i||A_j|D_B(d(A_i,A_j)\|d_{ij})\Bigg).
\end{split}
\end{equation}
Applying Proposition \ref{hgratedomprop} to each term of both outer
sums now yields the claim, because
\begin{equation*}
\sum_{i=1}^k(M(|\eta\cap A_i|)-1)+\sum_{i<j}(|\eta\cap A_i||\eta\cap A_j|-1)
=M(|\eta|)-M(k).
\end{equation*}
\end{proof}

\begin{remark}
It is rather surprising that the stochastic bound
\eqref{genrefinementgainub} depends only on the number of blocks in
$\eta$ --- not on their relative sizes, nor on the overall model size
$n$.
\end{remark}

\begin{proposition}
\label{refinementprop}
For any positive integer $m>k$, the following holds with high
probability:
\begin{equation*}
L(G_n|\xi_n)-\min_{\eta\ge\xi_n,\ |\eta|\le m}L(G_n|\eta)\le(m+M(m-k))\log n,
\end{equation*}
where the relation $\eta\ge\xi_n$ means that $\eta$ is a refinement of
$\xi_n$, and $M(\cdot)$ was defined in \eqref{Mdef}.
\end{proposition}
\begin{proof}
Let $\eta$ be a refinement of $\xi_n$. Refining the partition
w.r.t.\ $\xi_n$ yields a gain, based on the concavity of $H$, in the
code part $L_4+L_5$, but costs in the parts $L_1$, $L_2$ and $L_3$. We
have to relate these to each other.

Consider first the value of $L_2(G_n|\eta)$. In our analysis, it is
important to distinguish between `large' and `tiny' blocks, where the
relative sizes of large blocks exceed some pre-defined number
$\epsilon$ and the rest can be arbitrarily small, even
singletons. Now, each block $A_i\in\xi_n$ must contain at least one
block $B_i\in\eta\cap A_i$ such that $|B_i|\ge|A_i|/m$. Define
\begin{equation*}
\epsilon:=\min_j|A_j|/(2mn).
\end{equation*}
Because no concavity gain can be obtained with an index pair
$\set{i,j}$ such that $d_{ij}\in\set{0,1}$, it does not restrict
generality to assume that $d_{ij}\in(0,1)$ for all $i,j$. Then, by \eqref{bmcodexi},
\begin{align*}
L_2(G_n|\eta)
&=\sum_{i=1}^k\Bigg(\sum_{B\in\eta\cap A_i}\log\left(\binom{|B|}{2}d_{ii}\right)
  +\frac12\sum_{
\begin{array}{c}
\scriptstyle      B,B'\in\eta\cap A_i\\
\scriptstyle      B\not=B'
    \end{array}
}\log(|B||B'|d_{ii})\Bigg)\\
&\quad+\sum_{i<j}\Bigg(\sum_{B\in\eta\cap A_i}\sum_{B'\in\eta\cap A_j}
  \log(|B||B'|d_{ij})\Bigg)\\
&\ge\sum_{i=1}^k\Bigg(\log\left(\frac{\epsilon^2n^2}{2}d_{ii}\right)
  +(|\eta\cap A_i|-1)\log(\epsilon n d_{ii})\Bigg)\\
&\quad+\sum_{i<j}\left(\log\left(\epsilon^2n^2d_{ij}\right)+
  (|\eta\cap A_i|+|\eta\cap A_j|-2)\log(\epsilon n d_{ij})\right)\\
&=(k|\eta|+k)\log n+c_1(D,\epsilon).
\end{align*}
On the other hand, the bound $|A_i|\le n$ yields
\begin{equation*}
L_2(G_n|\xi_n)\le(k^2+k)\log n+c_2(D).
\end{equation*}
Thus,
\begin{equation}
\label{L2diff}
L_2(G_n|\eta)-L_2(G_n|\xi_n)\ge k(|\eta|-k)\log n+c_1(D,\epsilon)-c_2(D).
\end{equation}
We obviously have also $L_1(G_n|\eta)\ge L_1(G_n|\xi_n)$, but this
difference is insignificant in the present context.

The refinement gain in code parts $L_4$ and $L_5$ was bounded in
Proposition \ref{genrefinementprop} stochastically by {\it Exp}$(1)$
random variables. The rate function (see the beginning of Appendix
\ref{chernoffsec}) of the distribution {\it Exp}$(1)$ is
\begin{equation*}
  I_E(x)=x-1-\log x.
\end{equation*}
Denote
\begin{equation*}
\mathrm{Gain}_n(\eta)
:=L_4(G_n|\xi_n)+L_5(G_n|\xi_n)-(L_4(G_n|\eta)+L_5(G_n|\eta)).
\end{equation*}
Proposition \ref{genrefinementprop} yields, using \eqref{samplerate} and Proposition \ref{chernoffprop}, that for $y>\log 2$
\begin{equation*}
  \begin{split}
&    \pr{\mathrm{Gain}_n(\eta)>y}\\
&\le\pr{\sum_{j=1}^{M(|\eta|)-M(k)}Y_i>y-(M(|\eta|)-M(k))\log 2}\\
&\le\exp\left(-(M(|\eta|)-M(k))
  I_E\left(\frac{y-(M(|\eta|)-M(k))\log 2}{M(|\eta|)-M(k)}\right)
\right)\\
&\le\exp\big(-y+(M(|\eta|)-M(k))\log y\big)
  \left(\frac{2e}{M(|\eta|)-M(k)}\right)^{M(|\eta|)-M(k)},
  \end{split}
\end{equation*}
where the second factor is bounded and will be henceforth neglected.

For two refinements of $\xi_n$, write $\eta'\sim\eta$ if the block
sizes of $\eta'$ in each $A_i$ are identical to those of $\eta$. The
number of refinements $\eta'$ of $\xi_n$ with $\eta'\sim\eta$ is
upperbounded by
\begin{equation*}
\exp\left(\sum_{A\in\xi_n}|A|H(\eta\cap A)\right)=e^{nH(\eta|\xi_n)},
\end{equation*}
where $H(\eta\cap A)$ denotes the entropy of the partition of $A$
induced by $\eta$. On the other hand, we have
\begin{equation*}
L_3(G_n|\eta)-L_3(G_n|\xi_n)=nH(\eta)-nH(\xi_n)=nH(\eta|\xi_n).
\end{equation*}

Write 
\begin{equation*}
  \begin{split}
\Delta L_{123}(\eta)&=L_1(G_n|\eta)+L_2(G_n|\eta)+L_3(G_n|\eta)\\
  &\quad-(L_1(G_n|\xi_n)+L_2(G_n|\xi_n)+L_3(G_n|\xi_n)),\\
\Delta L(\eta)&=L(G_n|\eta|)-L(G_n|\xi_n|).
  \end{split}
\end{equation*}
Denote $z(\eta):=(|\eta|+M(|\eta|-k))\log n$. Recalling
\eqref{L2diff}, the union bound yields
\begin{equation*}
  \begin{split}
    &\pr{\sup_{\eta'\sim\eta}\mathrm{Gain}_n(\eta')>\Delta L_{123}(\eta)+z(\eta)}\\
&\le\exp\left(nH(\eta|\xi_n)-\Delta L_{123}(\eta)-z(\eta)+(M(|\eta|)-M(k))
  \log(\Delta L_{123}(\eta)+z(\eta))\right)\\
&\le\exp\bigg(-z(\eta)+\big(-k(|\eta|-k)+M(|\eta|)-M(k)\big)\log n
  +c_3(D,\epsilon)\bigg).
  \end{split}
\end{equation*}
The number of different block size sequences
$\ell_1\ge\cdots\ge\ell_{|\eta|}$ is upper bounded by
$n^{|\eta|-1}=e^{(|\eta|-1)\log n}$. Thus, a second application of the
union bound yields
\begin{equation*}
  \begin{split}
    &\pr{\sup_{\eta\ge\xi_n,\ |\eta|\le m}\Delta L(\eta)>z(\eta)}\\
%&\le me^{\big(-\big[k(|\eta|-k)\big]+\big[M(|\eta|)-M(k)\big]
%-\big[|\eta|+M(|\eta|-k)\big]+\big[|\eta|-1\big]\big)\log n
%+\mathit{const}}\\
&\le me^{\big(-\big[k(m-k)\big]+\big[M(m)-M(k)\big]
-\big[m+M(m-k)\big]+\big[m-1\big]\big)\log n
+\mathit{const}}\\
&=\frac{\mathit{const}}{n}\to0,
  \quad\mbox{as }n\to\infty.
  \end{split}
\end{equation*}
\end{proof}

\begin{corollary}
\label{concacor}
$\xi_n$ is the unique minimizer of $L(G_n|\eta)$ among
partitions $\eta$ with $|\eta|=k$.
\end{corollary}
\begin{proof}
By Proposition \ref{neartoendprop},
\begin{equation*}
\min_{|\eta|=k,\ \mathfrak{d}(\eta,\xi_n)\le\epsilon_0}L(G_n|\eta)>L(G_n|\xi_n)
\end{equation*}
with high probability. On the other hand, Proposition \ref{bigdifprop}
yields that, with high probability,
\begin{equation*}
\min_{|\eta|=k,\ \mathfrak{d}(\eta,\xi_n)>\epsilon_0}L(G_n|\eta)>
  L(G_n|\xi_n\vee\eta)+c_{\epsilon_0}n^2.
\end{equation*}
By Proposition \ref{refinementprop},
\begin{equation*}
\min_{|\eta|=k}L(G_n|\xi_n\vee\eta)>L(G_n|\xi_n)-(k+M(k^2-k))\log n
\end{equation*}
with high probability. It remains to note that $n^2$ grows faster than
$\log n$.
\end{proof}

\begin{proposition}
\label{bigblockprop}
Let $\epsilon\in(0,\min_i\gamma_i)$. Consider refinements $\eta$ of
$\xi_n$ with relative minimal block size $\epsilon$, i.e.\ the set
\begin{equation}
\mathcal{B}^{(n)}_\epsilon=\eset{\eta\ge\xi_n}{|B|\ge n\epsilon\ \forall B\in\eta}.
\end{equation}
With high probability,
\begin{equation}
\min_{\eta\in\mathcal{B}^{(n)}_\epsilon\setminus\set{\xi_n}}L(G_n|\eta)>L(G_n|\xi_n).
\end{equation}
\end{proposition}
\begin{proof}
The restriction $\eta\in\mathcal{B}^{(n)}_\epsilon$ implies
$|\eta|\le\lfloor1/\epsilon\rfloor=:m$. The difference to the
conditions of Proposition \ref{refinementprop} is now just the
magnitude of $L_2(\eta)$. When $\eta\in\mathcal{B}^{(n)}$,
\begin{equation*}
L_2(G_n|\eta)\ge M(|\eta|)\log n^2+c(D,\epsilon,|\eta|),
\end{equation*}
so that 
\begin{equation*}
L_2(G_n|\eta)-L_2(G_n|\xi_n)\ge 2(M(|\eta|)-M(k))\log n+c(D,\epsilon,|\eta|).
\end{equation*}
By a corresponding computation as in the proof of Proposition
\ref{refinementprop}, we obtain for any fixed
$\eta\in\mathcal{B}^{(n)}_\epsilon$ that
\begin{equation*}
  \begin{split}
    &\pr{\sup_{\eta'\sim\eta}\mathrm{Gain}_n(\eta')
  >\Delta L_{123}(\eta)-\frac12\log n}\\
&\le\exp\left(-(M(|\eta|)-M(k))\log n+\mathit{const}\right).
  \end{split}
\end{equation*}
Proceeding with a second union bound like in the proof of Proposition
\ref{refinementprop} yields
\begin{equation*}
  \begin{split}
&\pr{\min_{\eta\in\mathcal{B}^{(n)}_\epsilon\setminus\set{\xi_n}}L(G_n|\eta)
\le L(G_n|\xi_n)+\frac12\log n}\\
&\le\max_{q\le m}\,\exp\left(\left(\bigg[-(M(q)-M(k))\bigg]
  +\bigg[q-1\bigg]+\frac12\right)\log n+\mathit{const}\right)\\
&=\exp\left(-\frac12\log n+\mathit{const}\right)\\
&=\frac{\mathit{const}}{\sqrt{n}}\to0,
  \quad\mbox{as }n\to\infty,
  \end{split}
\end{equation*}
where the maximum over $q$ was obtained with the smallest value $q=k+1$.
\end{proof}
\section{Algorithms, codes and illustrations}\label{algoritms}
\subsection{Regular decomposition algorithms based on MDL}\label{algosec}

In this section we present algorithms that we have used in actual
computations of regular decompositions of graph and matrix data. These
are written for standard two-part MDL, where the code lengths $L_4$ and
$L_5$ have a usual interpretation as a minus log-likelihood of a graph
corresponding to a stochastic block model.

Thus, we use link coding lengths found in the upper bound of
Proposition \ref{twopartprop}. In many cases, this is all that can be
computed realistically. Moreover such overestimating is not critical since  the over-fitting seems to be a common problem, because the minimum of MDL tend to be very shallow and is easily passed unnoticed. We can obviously describe a partition into $k$
nonempty sets using an $n\times k$ binary matrix with all row sums
equal to one and requiring that none of the column sums equals
zero. The space of all such matrices we denote as ${\cal{R}}_k$ and
the members of this set as $R\in{\cal{R}}_k$.

\begin{definition}
For a given graph $G\in\Omega_n$ with adjacency matrix $A$ and a
partition matrix $R\in{\cal{R}}_k$, denote
\begin{equation*}
P_1(R):= R^TAR,
\end{equation*}
where $\cdot^T$ stands for matrix transpose, the column sums of $R$
are denoted as 
\begin{equation*}
n_\alpha:=(R^TR)_{\alpha,\alpha},\quad 1\leq\alpha\leq k,
\end{equation*}
and the number of links within each block and between block pairs as
\begin{equation*}
e_{\alpha,\beta}(R)= (1-\frac{1}{2}\delta_{\alpha,\beta})(P_1(R))_{\alpha,\beta},
\end{equation*}
where $\delta_{\alpha,\beta}=1$ if $\alpha=\beta$ and $\delta_{\alpha,\beta}=0$ otherwise. 
Then define
\begin{equation*}
(P(R))_{\alpha,\alpha}:=\indic{n_\alpha>1}\frac{e_{\alpha,\beta}(R)}{\binom{n_\alpha}{2}},\quad
(P(R))_{\alpha,\beta}:=\frac{e_{\alpha,\beta}(R)}{n_\alpha n_\beta},\ \alpha\not=\beta.
\end{equation*}
%To avoid exactly $0$ or $1$ elements we define finally: 
%\begin{equation*}
%(P(R))_{\alpha,\beta}:=\min(\max((P_2(R))_{\alpha,\beta},\delta),1-\delta).
%\end{equation*}
\end{definition}

Then the coding length of the graph corresponding to $A$ using the model $R$ is:

\begin{definition}
\label{algcodelendef}
\begin{align*} 
l_k(G(A)\mid R\in{\cal R}_k):=& \sum_{1\leq i< j\leq k}n_in_j H((P(R))_{i,j})+\sum_{1\leq i\leq k}\binom{n_i}{2}H((P(R))_{i,i}) &  
\\&+ l_k(R),
\end{align*}
where 
$$
l_k(R)= \sum_{1\leq i\leq k}n_iH(n_i/n)+ \sum_{1\leq i\leq j\leq k}l^*(e_{i,j}(R))
$$ 
is the code length of the model, according to our theory and notation.
\end{definition}

The {\bf two-part MDL program} of finding the optimal model, denoted
as $R_{k^*}$, can now be written as:
\begin{eqnarray}
\boxed{(k^*,R_{k^*})
:=\argomin_{1\leq k\leq n}\min_{R\in{\cal R}_k }l_k (G(A)\mid R\in{\cal R}_k)}
\end{eqnarray}
To solve this program approximately, we can use the following greedy algorithm.

\begin{algorithm}
{\underline {Greedy Two-part MDL}}\\
{\bf Input:} $G=G(A)\in\Omega_n$ a simple graph of size $n$.\\
{\bf Output:} $(k^*,R_{k^*}\in {\cal R}_{k^*})$, such that the two-part code for $G$ is shortest possible for all models in ${\cal M}_n$ by using this pair as a model.\\
Start: $k=1$, $l^*=\infty$, $R\in{\cal R}_1=\set{I}$, 
$k^*=1$, where $I$ is denotes the $n\times 1$  matrix with all elements equal to $1$.\\
{\bf 1. Find} 
$$
\hat{R}_k(G):=\argomin_{R\in {\cal R}_{k}} (l_k (G\mid R) 
$$
using subroutine {\bf ARGMAX k} (Algorithm \ref{argmaxk}).
\newline
{\bf 2. Compute } $l_k(G)= \lceil l_k(G\mid \hat{R}(G))\rceil+l_k (\hat{R}(G))$\\
{\bf 3.  If} $l_k(G)< l^*$  then $ l^*=l_k(G)$, $R_{k^*}=\hat{R}_k(G)$ , $k^*=k$\\
{\bf 4.} $k=k+1$\\
{\bf 5. If } $k>n$, {\bf Print} $(R_{k^*},k^*)$ and {\bf STOP} the program. \\
{\bf 6.} {\bf GoTo 1}.
\end{algorithm}

\begin{definition}
Matrices noted as $LogP(R)$ and $Log(1-P(R))$ are defined as:
\begin{align*}
(LogP(R))_{\alpha,\beta}&:=\log(P(R)_{\alpha,\beta}),\\ 
(Log(1-P(R)))_{\alpha,\beta}&:=\log(1-P(R)_{\alpha,\beta}),\\
 L(R)&:= -ARLogP(R)-(1-A) R Log(1-P(R)),
\end{align*}
where we set $\log0=0$ (all $\log 0$ values will be later multiplied by 0),
\begin{equation*}
\beta(i,R):= \inf\{\beta:\beta=\argomin_{1\leq\alpha\leq k}(L(R))_{i,\alpha}\}, 1\leq i\leq n.
\end{equation*}
A mapping $\Phi:{\cal{R}}_k\rightarrow {\cal{R}}_k$ is defined as follows:
$$
\Phi(R)_{i,\alpha}= \delta_{\alpha,\beta(i,R)}.
$$
\end{definition}

The mapping $\Phi(R)$ moves each node to a possibly different block
in such a way that the description length would be minimized if all other nodes
stay in their current blocks. A Python code with generating synthetic binary graphs is given in \cite{reittugithub}.
\begin{algorithm}
\label{argmaxk}  
\underline{ARGMAX k}
\newline 
Algorithm for finding optimal regular decomposition for fixed $k$\\
{\bf Input}: $A$: the adjacency matrix of a graph (an $n\times n$ symmetric binary matrice with zero trace); $N$: an integer (the number of iterations in the search of a global optimum); 
%$\delta$: a small positive real number that is selected so small that the output of this algorithm does not depend on its actual value; 
$k$: a positive integer.\\
Start: $m=1$.\\
{\bf 1.}  $i:=0$; generate a uniformly random element $R_i\in\mathcal{R}_k$.\\
{\bf 2.}  If at least one of the column sums of $R_i$ is zero, {\bf GoTo 1}. Otherwise, set
$$
R_{i+1}:=\Phi(R_i).
$$
{\bf 3.} {\bf If}  $R_{i+1}\neq R_i$, set $i:=i+1$ and {\bf GoTo 2}.\\
{\bf 4.} $R(m):=R_{i}$; $m=m+1$; 
$l(m):= \sum_{i=1}^n \min_{1\leq\alpha\leq k}(L(R(m)))_{i,\alpha}$.\\
{\bf 5.} If $m<N$, {\bf GoTo 1}. \\
{\bf 6.} $M:=\{m: l(m)\leq l(i),i=1,2,...,N\}$; $m^*:=\inf M$. \\
{\bf Output} optimal solution: $R(m^*)$.
\end{algorithm}

For very large graphs, the program may not be solvable in the sense
that it is not possible and reasonable to go through all possible
values of $k\in\{1,2,\cdots, n\}$.  One option is to limit the range
of $k$. In case that no minimum is found, then use as an optimal
choice the model found for the largest $k$ within this range. Another
option is to find the first minimum with smallest $k$ and stop. When
the graph is extremely large, it makes sense to use only a randomly
sampled sub-graph as an input --- indeed, when $k^*<<n$, a
large-scale structure can be estimated from a sample \cite{reittuetalljournal}.

Our algorithm for Poissonian block models and matrices is essentially
similar, with certain differences in formulae as detailed
below. Algorithms for other cases like directed graphs and
non-quadratic matrices are written very similarly, although two
partitions are needed, one for rows and one for columns. The logic of
the solution remains the same however.

A semi-heuristic two-part MDL algorithm for finding a regular
decomposition for an $n\times m$ matrix $A$ with non-negative entries
works as follows. The decomposition takes a form of a
bi-clustering: there are two partitions, one for rows and one for
columns. Such partitions are described by binary matrices with row
sums equal to one. The two-part MDL program to find an optimal regular
decomposition is written as follows:

The row partition-matrices are denoted as $R\in{\cal R}_{k_1}$ with
dimensions $n\times k_1$, $1\leq k_1\leq n$, and the column partition
matrices as $C\in{\cal C}_{k_2}$ with dimensions $m\times k_2$, $1\leq
k_2\leq m$.

Let us formulate the cost function for the matrix case that is derived from Eq (9). The number of
matrix elements in row group $\alpha$ and column group $\beta$ can be
written as a matrix element:
$$
(N)_{\alpha, \beta}= (R^TR)_{\alpha,\alpha}(C^TC)_{\beta,\beta}:= n_\alpha m_\beta.
$$ 
Assuming that all blocks are non-empty, all $N_{\alpha,\beta}>0$,
we can define an average matrix element of block $\alpha,
\beta$. First compute the sum of all matrix elements of $A$ over such
a block:
$$
e_{\alpha,\beta}= (R^TAC)_{\alpha,\beta}.
$$
The corresponding block averages form a $k_1\times k_2$ $P$-matrix with elements 
$$
(P)_{\alpha,\beta}= \frac{e_{\alpha,\beta}}{(N)_{\alpha,\beta}}.
$$
The coding length of the matrix $A$ using a two-part MDL code with partitions $(R,C)$ can be written as
\begin{align*}
&l_{k_1,k_2}(G(A)\mid R\in{\cal R}_{k_1},C\in{\cal C}_{k_2}))\\
&=\sum_{1\leq \alpha\leq k_1,1\leq \beta\leq k_2}\big\{e_{\alpha,\beta}(1-\log( (P)_{\alpha,\beta})+l^* ([e_{\alpha,\beta}])\big\}\\
&\quad+ \sum_{1\leq \alpha\leq k_1}n_\alpha H\big(\frac{n_\alpha}{n}\big)+\sum_{1\leq \beta\leq k_2}m_\beta H\big(\frac{m_\beta}{m}\big)+ k_1 \times k_2 c.
\end{align*}
Here we assume a similar handling of $\log 0$'s as in the binary
case. $n_\alpha$ and $m_\beta$ are the sizes of row and
column blocks, and $[e_{\alpha,\beta}]$ denotes the integer part of
$e_{\alpha,\beta}$; it is assumed that such block sums are large
numbers with finite decimal precision $c>0$. The description length of
such decimals is the last term and it is small compared with other terms for
large matrices and can be safely ignored. Similarly to the binary
case, the two-part MDL program is defined as:
\begin{definition}
\label{MDL matrix}  
\begin{equation}
%\boxed{
  \begin{split}
&(k_1^*, k_2^*,R,C)\\
&:=\argomin_{( k_1, k_2)}\left(\min_{R\in{\cal R}_{k_1} ,C\in{\cal C}_{k_2}}l_{k_1,k_2}(G(A)\mid R\in{\cal R}_{k_1},C\in{\cal C}_{k_2})\right)  
  \end{split}
%  }
\end{equation}
where
$$
1\leq k_1\leq n,\quad 1\leq k_2\leq m.
$$
\end{definition}
The greedy algorithm for solving this program is very similar to the
case of a binary matrix. The difference is that two parametric sequences
of partitions must be searched, ${\cal R}_i$ and ${\cal C}_j$, and in
the subroutine that finds the optimal partitions for fixed $i$ and
$j$. One may consider different strategies in the corresponding search
of an optimal pair. For instance, moving first along the diagonal $i=j$,
and finding the value where the cost function (coding length of $A$) has a
knee-point, and after that make an off-diagonal search near that value.
Another option could be moving along the steepest descent direction of the
cost function, or alternating the directions of increments in $i$ and $j$,
until a saturation is reached in one direction, and then keeping that
parameter fixed and finding the optimum on the second parameter. This question
is a subject to further experimenting with real and artificial
data. Therefore we write only the subroutine that finds the optimal
partitioning with fixed $k_1$ and $k_2$ in a greedy fashion. First we
need

\begin{definition}
Define the mappings $\Phi_R:{\cal{R}}_{k_1}\times {\cal
  C}_{k_2}\rightarrow {\cal{R}}_{k_1}$ and
$\Phi_C:{\cal{R}}_{k_1}\times {\cal C}_{k_2}\rightarrow
{\cal{C}}_{k_2}$ as follows. Let $E$ be an $n\times m$
matrix with all elements equal to $1$. Then, using the definition of
$P$-matrix and $LogP$-matrix (related to $P$ as in binary case),
define two matrices using a block-matrix notation:
$$
L(R,C)=
\begin{pmatrix}
E & A 
\end{pmatrix}
\begin{pmatrix}
 CP^T\\-C(LogP)^T
\end{pmatrix}
$$
and
$$
M(R,C)=
\begin{pmatrix}
E^T & A^T 
\end{pmatrix}
\begin{pmatrix}
 RP\\-RLogP
\end{pmatrix}
.
$$
Define
$$
\beta_1(i,R,C)= \inf\{\beta:\beta=\argomin_{1\leq\alpha\leq k_1}(L(R,C))_{i,\alpha}\},\quad 1\leq i\leq n,
$$
and 
$$
\beta_2(i,R,C)= \inf\{\beta:\beta=\argomin_{1\leq\alpha\leq k_2}(M(R,C))_{i,\alpha}\},\quad 1\leq i\leq m.
$$
Then,
$$
\Phi_R(R,C)_{i,\alpha}= \delta_{\alpha,\beta_1(i,R,C)},\quad 1\leq\alpha\leq k_1,\quad 1\leq i\leq n,
$$
and 
$$
\Phi_C(R,C)_{i,\alpha}= \delta_{\alpha,\beta_2(i,R,C)},\quad 1\leq\alpha\leq k_2,\quad 1\leq i\leq m.
$$
\end{definition}
The main greedy subroutine is:
\begin{algorithm}
\label{argmax k1k2}  
\underline{ARGMAX $(k_1,k_2)$}
\newline 
Algorithm for finding optimal regular decomposition for fixed $(k_1,k_2)$.\\
{\bf Input}: $A$: a real $n\times m$ matrix with non-negative entries; $N$: a positive integer (the number of iterations in the search of a global optimum); $(k_1,k_2)$: a pair of positive integers.\\
Start: $m=1$.\\
{\bf 1.}  $i=0$; generate uniformly random elements $R_i\in{\cal R}_{k_1}$. and $C_i\in{\cal C}_{k_2}$.\\
{\bf 2.}  If at least one of column sums of $R_i$ or $C_i$ is zero, {\bf GoTo 1}. Otherwise, set
$$
R_{i+1}:=\Phi_R(R_i,C_i),\quad C_{i+1}:=\Phi_C(R_i,C_i).
$$
{\bf 3.} {\bf If}  $R_{i+1}\neq R_i$ or $C_{i+1}\neq C_i$, set $i:=i+1$ and {\bf GoTo 2}.\\
{\bf 4.} $R(m):= R_{i}$; $C(m):=C_i$; $m:=m+1$;
 \begin{align*}
e_{\alpha,\beta}&:= \left(T^T(m)AC(m)\right)_{\alpha,\beta},\quad &1\leq\alpha\leq k_1,\quad 1\leq\beta\leq k_2;\\ 
N_{\alpha,\beta}&:=\left(R^T(m)R(m)\right)_{\alpha,\alpha}\left(C^T(m)C\right)_{\beta,\beta},
\quad &1\leq\alpha\leq k_1,\quad 1\leq\beta\leq k_2;\\
l(m)&:= \sum_{1\leq\alpha\leq k_1, 1\leq\beta\leq k_2} e_{\alpha,\beta}\left(1-\log\frac{e_{\alpha,\beta}}{N_{\alpha,\beta}}\right).
\end{align*}
{\bf 5.} If $m<N$, {\bf GoTo 1}. \\
{\bf 6.} $M:=\{m: l(m)\leq l(i),i=1,2,...,N\}$; $m^*:=\inf M$. \\
{\bf OUTPUT} optimal solution: $(R(m^*),C(m^*))$.
\end{algorithm}

\begin{remark}
It is also possible to find regular decompositions in the case of
partly missing matrix elements \cite{reittuetall, reittuetalljournal}and also in the case of mixed positive
and negative entries. In the latter case, we can use the idea of
directed links already used in\cite{tusnady}. In the first case, we
note that the main characteristic of the regular decomposition is the
$P$-matrix with elements that are averages of the data matrix over
large blocks that can be estimated, in many cases, despite a
portion of data is missing.
\end{remark}
  \subsection{Simulations}

To illustrate MDL based RD, we run several computer experiments. The purpose was to verify whether the right numbers of blocks could be found.

We used a simple stochastic block model for this purpose. There were $k=5$ blocks and each node was placed in a block uniformly at random. The number of nodes was $1000$. Link probabilities in blocks and between them were uniformly random reals in range $[0,1]$. After the probabilities were generated they were multiplied by a constant factor $p$. These factors had values $0.1, 0.3, 0.5, 0.7,0.9$ and experiments were run for each value. The aim of the factor was to generate graphs with different densities.

For each set of parameters, experiments were repeated $8$ times to get diversity of samples. RD algorithm was run in one node of a high power computation cluster at VTT. The whole experiment took several days to complete. 

The result in most cases is that the right number of blocks can be identified. The results are presented graphically in the series of figures in Figure \ref{simresufig}. The cost function is the smallest found - log-likelihood + the complexity term $k\log n+ .5k(k+1)\log((n-k)(n-k+1))$. The first corresponds to the length of a code that encodes the partition and the last one to the coding length of the  number of links between and inside the block.  The cost function shows a characteristic kink at value $k=5$. After this value of $k$ is exceeded the cost function remains almost flat. 

\begin{figure}[htb]
\begin{center}
\begin{tabular}{cc}
%\\
\includegraphics[width=58mm]{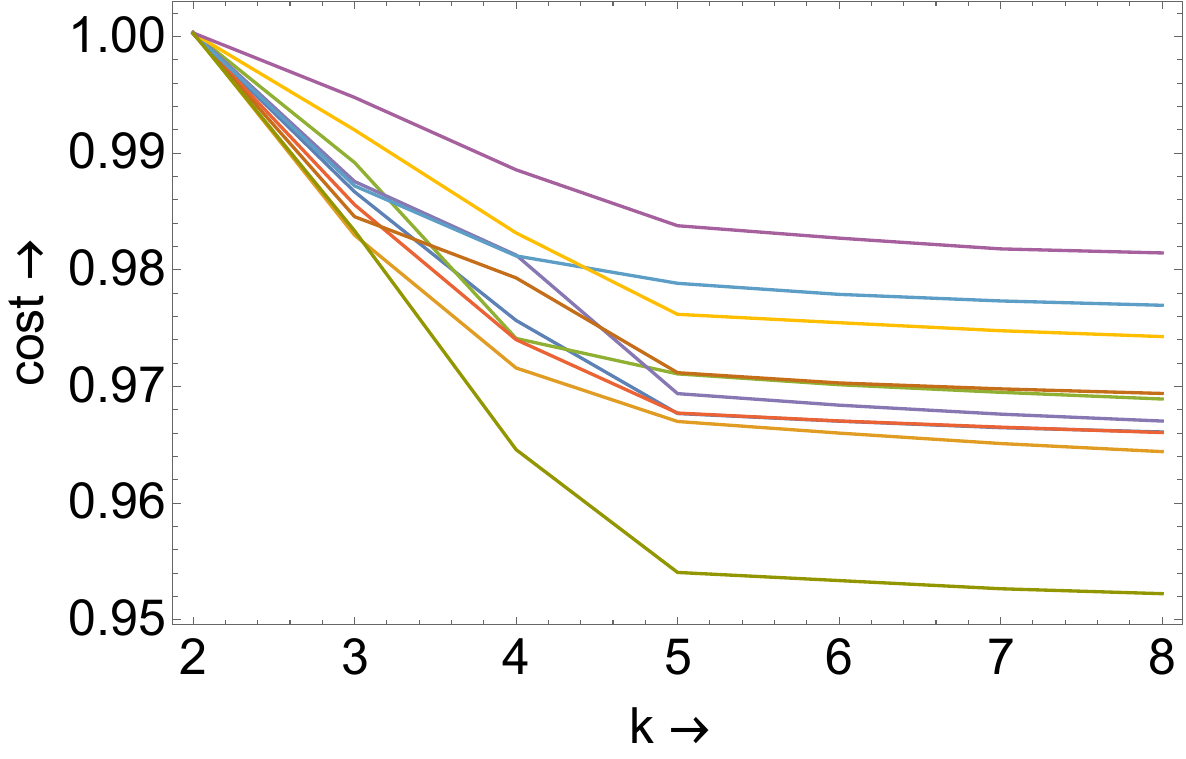}&\includegraphics[width=58mm]{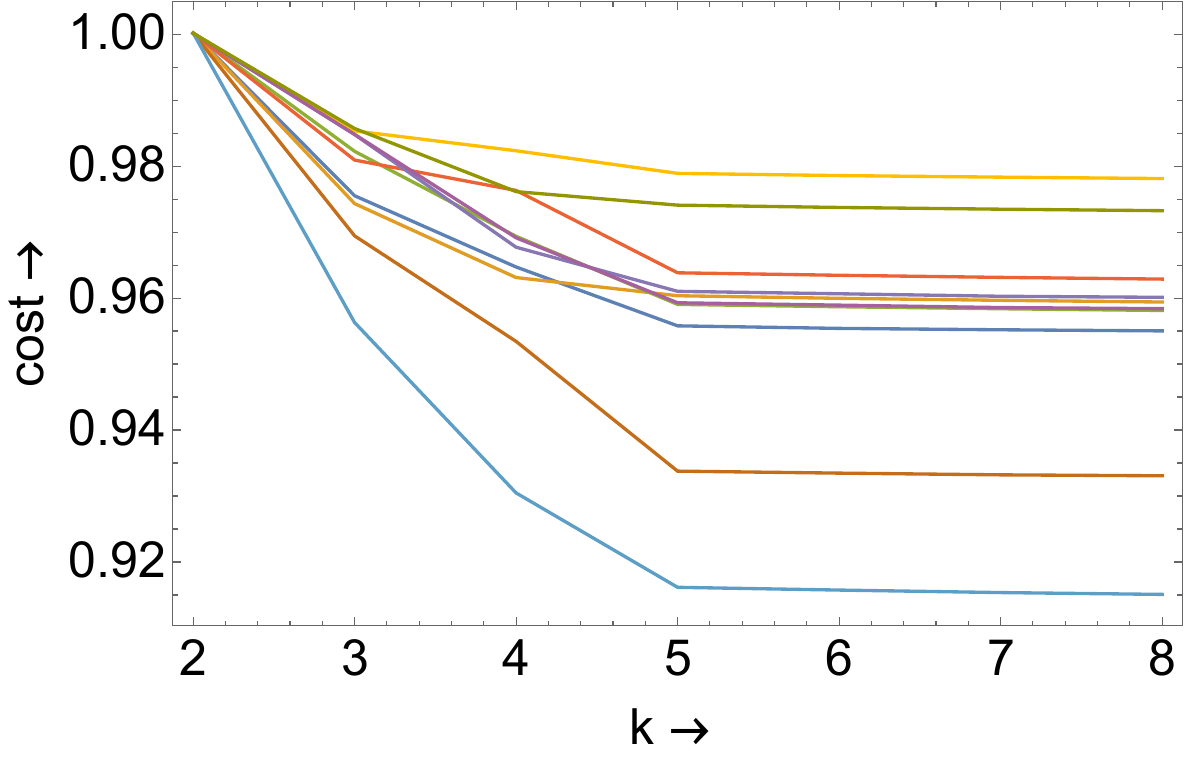}\\
\newline
\includegraphics[width=58mm]{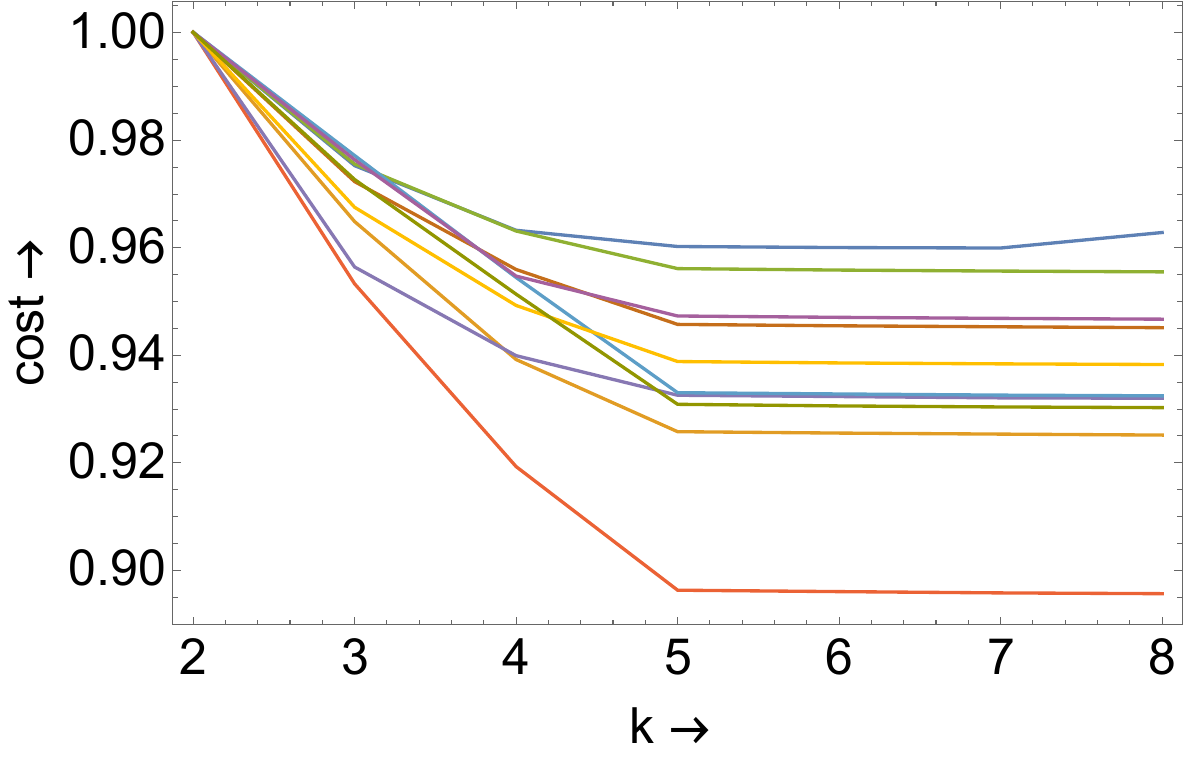}&\includegraphics[width=58mm]{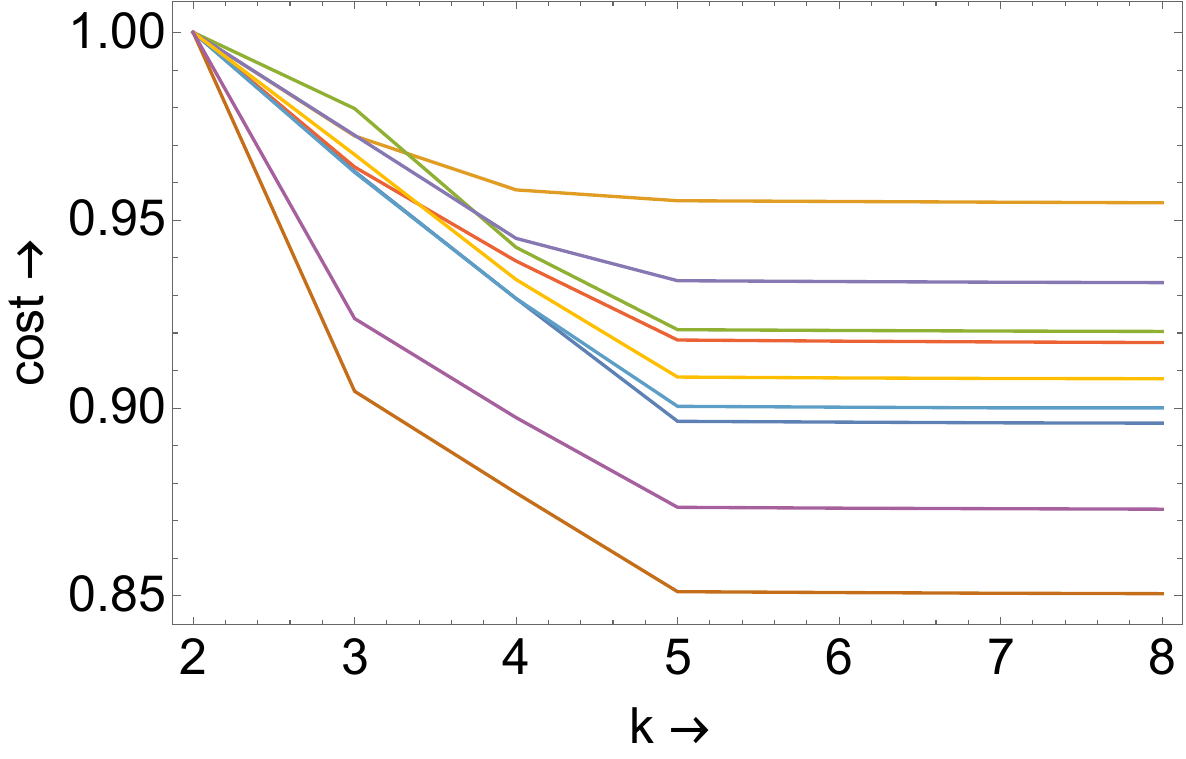}\\
\includegraphics[width=58mm]{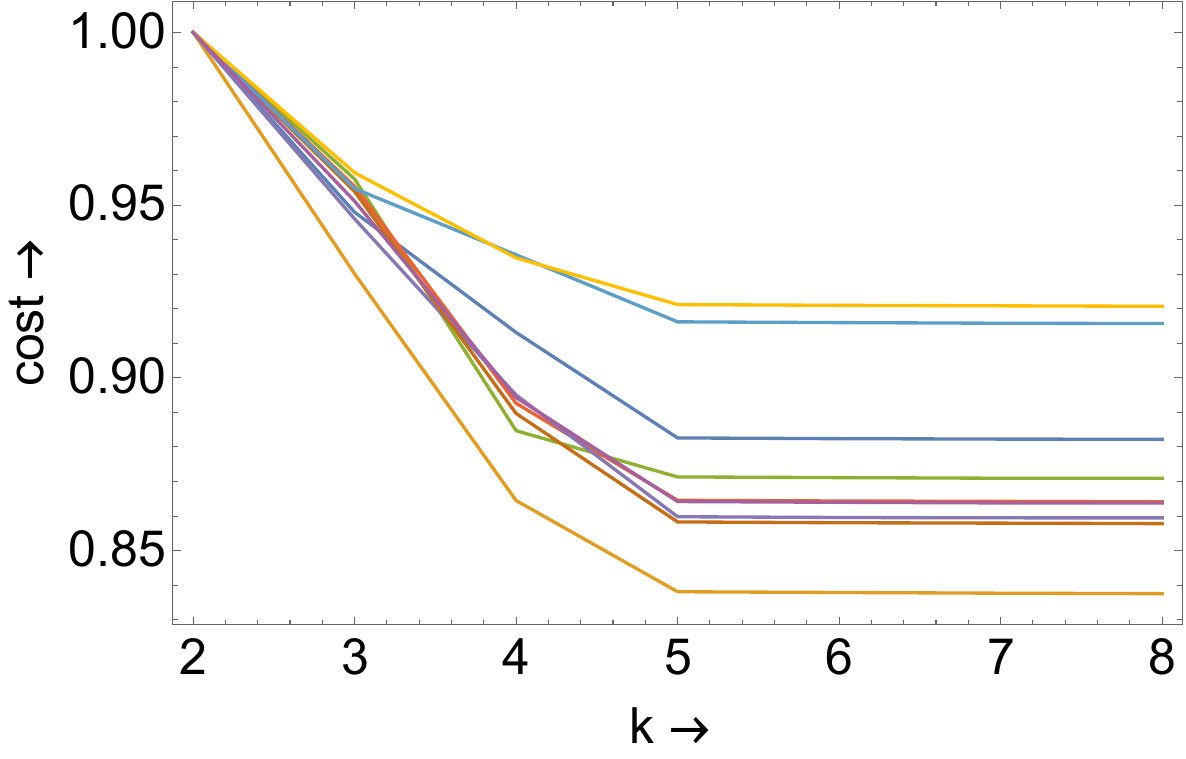}
\end{tabular}
\end{center}
\caption{Normalized cost functions as a function of number of blocks $k$. From top left to down expected density of the graph grows according to factor $p$ in each figure. Each line corresponds to an independent repetition of experiments. In most cases $k=5$ is clearly identified as a value after which the function becomes almost flat. }
\label{simresufig}
\end{figure}

The main observation is that RD must be done carefully, meaning that the greedy algorithm must be rerun many times in order to find a good approximation of the global optimum. In our case $2000$ runs seems to work.

In \cite{reittuetalljournal} we showed that RD can work on very large graphs generated from a SBM, provided $k$ is a known constant and the graph is dense with fixed link probabilities independent on graph size. When MDL is used in RD, it is very likely that the right structure can be found from samples of a very large graph without knowing the right $k$.

SBM is called sparse when link probabilities tend to $0$ as $n\rightarrow \infty$. In this case we suggested \cite{reittudistance} to use RD to analyze the graph distance matrix instead of the adjacency matrix. In cases when the graph distances can be estimated, this could be used to find block structures of very large graphs using a similar sampling approach as in the dense graph case.  
%\newpage
\appendix

\section{Chernoff bounds and other information-theoretic preliminaries}\label{chernoffsec}

Consider a random variable $X$ with moment generating function
\begin{equation*}
\phi_X(\beta)=\E{e^{\beta X}},
\end{equation*}
and denote $\mathcal{D}_X=\eset{\beta}{\phi_X(\beta)<\infty}$. We
restrict to distributions of $X$ for which $\mathcal{D}_X$ is an open
(finite or infinite) interval. The corresponding {\em rate function}
is
\begin{equation*}
I_X(x)=-\inf_{\beta\in\mathcal{D}_X}(\log\phi_X(\beta)-\beta x).
\end{equation*}
$I_X(x)$ is a strictly convex function with minimum 0 at $\E{X}$ and
value $+\infty$ outside the range of $X$. For the mean $\bar{X}_n=\frac{1}{n}\sum_1^nX_i$ of
  i.i.d.\ copies of $X$, we have
\begin{equation}
\label{samplerate}
I_{\bar{X}_n}(x)=nI_X(x).
\end{equation}

The Chernoff bound (also known as Cram\'er-Lundberg bound):

\begin{proposition}
\label{chernoffprop}
\begin{align*}
  \pr{X<x}&\le e^{-I_X(x)}\quad\mbox{for }x<\E{X},\\
  \pr{X>x}&\ge e^{-I_X(x)}\quad\mbox{for }x>\E{X}.
\end{align*}
\end{proposition}

This has the following simple consequence that plays an important role
below. Let the convex hull of the support of $X$ be the closure of
$(x^-,x^+)$, and denote
\begin{equation*}
  a^-:=\lim_{x\downarrow x^-}I_X(x)=-\log\pr{X=x^-},\quad
  a^+:=\lim_{x\uparrow x^+}I_X(x)=-\log\pr{X=x^+}
\end{equation*}
($a^⁻<\infty$ if and only if the distribution of $X$ has an atom at
$x^-$, similarly for $a^+$). We can now write
\begin{equation*}
  \begin{split}
I_X(x)&=I^-_X(x)\indic{(x^-,\E{X}]}(x)+I^+_X(x)\indic{[\E{X},x^+)}(x)\\
    &\quad+a^-\cdot\indic{x^-}(x)+a^+\cdot\indic{x^+}(x)
    +\infty\cdot\indic{\Re\setminus[x^-,x^+]}(x).
      \end{split}
\end{equation*}
With the assumptions made above, the functions $I^-_X(x)$ and
$I^+_X(x)$ are, respectively, bijections from $(x^-,\E{X}]$ and
$[\E{X},x^+)$ to $[0,a^-)$ and $[0,a^+)$. 

\begin{lem}
\label{probofinfolemma}
\begin{align*}
 I_X(X)\buildrel{(st)}\over\le \log 2 +Y,
\end{align*}
where $\buildrel{(st)}\over\le$ denotes stochastic order and $Y$ is a
random variable with distribution {\it Exp}$(1)$.
\end{lem}
\begin{proof}
For any $z\ge0$, we have
\begin{align*}
\pr{I_X(X)>z}&=\pr{\indic{X<\E{X}}I_X(X)>z\mbox{ or }\indic{X>\E{X}}I_X(X)>z}\\
&=\pr{\indic{X<\E{X}}I^-_X(X)>z}+\pr{\indic{X>\E{X}}I^+_X(X)>z}\\
&=\indic{z<a^-}\pr{X<I^{-(-1)}_X(z)}+\indic{z<a^+}\pr{X>I^{+(-1)}_X(z)}\\
&\le\indic{z<a^-}\exp(-I_X(I^{-(-1)}_X(z)))+\indic{z<a^+}\exp(-I_X(I^{+(-1)}_X(z)))\\
&\le2e^{-z}\\
&=e^{-(z-\log 2)},
\end{align*}
where the first inequality comes from Proposition
\ref{chernoffprop}.
Thus,
\begin{equation*}
  \pr{I_X(X)>z}\le\min\left\{1,\,e^{-(z-\log 2)}\right\}=e^{-(z-\log 2)^+}
  =\pr{\log 2+Y>z}.
\end{equation*}
\end{proof}

In the case that $X$ has the {\it Bernoulli}$(p)$ distribution,
we have 
\begin{equation}
\label{Ixpdef}
I_X(x)=D_B(x\|p):=x\log\frac{x}{p}+(1-x)\log\frac{1-x}{1-p}.
\end{equation}

\begin{lem}
\label{infodifflem}
The first and second derivatives of the functions $H(x)$ and $x\mapsto D_B(x\|p)$ are
\begin{align}
\label{infodiffH}
H'(x)&=\log\frac{1-x}{x},&\quad
H''(x)&=-\frac{1}{x(1-x)},\\
\label{infodiffI}
D_B'(x\|p)&=H'(p)-H'(x),&\quad
D_B''(x\|p)&=\frac{1}{x(1-x)}.
\end{align}
Since $D_B(p\|p)=D_B'(p\|p)=0$ and $D_B''(p\|p)=-H''(p)$, we also have 
\begin{equation}
\label{HI}
H(q)-(H(p)+H'(p)(q-p))=-D_B(q\|p),
\end{equation}
and
\begin{align*}
&\lim_{n\to\infty}n\left[H\left((1-\frac{1}{n})p+\frac{1}{n}q\right)
-\left((1-\frac{1}{n})H(p)+\frac{1}{n}H(q)\right)\right]\\
&=(q-p)H'(p)-(H(q)-H(p))\\
&=D_B(q\|p).
\end{align*}
\end{lem}

\begin{proposition}
\label{binkulidentprop}
Let $n\ge2$ and let $X_1$ and $X_2$ be independent random variables
with distributions {\it Bin}$(m,p)$ and {\it Bin}$(n-m,p)$,
respectively. Denote $X_{12}=X_1+X_2$ and $\bar{X}_1=X_1/m$,
$\bar{X}_2=X_1/(n-m)$, $\bar{X}_{12}=X_{12}/n$. Then the following
identities hold:
\begin{align}
\label{binkulidentity0}
&mD_B(\bar{X}_1\|p)+(n-m)D_B(\bar{X}_2\|p)-nD_B(\bar{X}_{12}\|p)\\
\label{binkulidentity1}
&=X_{12}D_B\left(\left.\frac{X_1}{X_{12}}\right\|\frac{m}{n}\right)
  +(n-X_{12})D_B\left(\left.\frac{m-X_1}{n-X_{12}}\right\|\frac{m}{n}\right)\\
\label{binkulidentity2}
&=mD_B(\bar{X}_1\|\bar{X}_{12})
  +(n-m)D_B\left(\left.\frac{X_{12}-X_1}{n-m}\right\|\bar{X}_{12}\right).
\end{align}
\end{proposition}

The identities in Proposition \eqref{binkulidentprop} are obtained by
writing the full expression of \eqref{binkulidentity0} and
re-arranging the $\log$ terms in two other ways. Formulae
\eqref{binkulidentity1} and \eqref{binkulidentity2} are written
without $X_2$, expressing the fact that any two of the three random
variables $X_1$, $X_2$ and $X_{12}$ contain same information as the
full triple. Note that \eqref{binkulidentity1} and
\eqref{binkulidentity2} do not contain $p$. This reflects the fact
that the conditional distribution of $X_1$ given $X_{12}$, known as
the hypergeometric distribution, does not depend on $p$. The identity
of \eqref{binkulidentity0} and \eqref{binkulidentity2} can be
interpreted so that the two positive terms of \eqref{binkulidentity0}
measure exactly same amount of information about $p$ as what is
subtracted by the negative term. Moreover, \eqref{binkulidentity2} has
the additional interpretation of presenting the rate function of the
hypergeometric distribution:

\begin{proposition}
\label{hypergrateprop}
Let $X$ have the distribution {\it Hypergeometric}$(n,m,z)$,
i.e.\ the conditional distribution of $X_1$ of Proposition
\ref{binkulidentprop} given that $X_{12}=z$. The rate function of $X$ is
\begin{equation}
\label{hypergrateq}
I_X(x)=mD_B\left(\left.\frac{x}{m}\right\|\frac{z}{n}\right)
  +(n-m)D_B\left(\left.\frac{z-x}{n-m}\right\|\frac{z}{n}\right).
\end{equation}
\end{proposition}
\begin{proof}
Define the bivariate moment-generating function of $(X_1,X_2)$
\begin{equation*}
\phi(\alpha,\beta)=\E{e^{\alpha X_1+\beta X_2}}.
\end{equation*}
Write
\begin{equation*}
\cPr{X_1=m}{X_{12}=z}=\frac{\pr{X_1=m,\ X_2=z-m}}{\pr{X_{12}=z}},
\end{equation*}
and note that we can assume $p=z/n$. We can now derive the claim using
$\phi(\alpha,\beta)$ in similar manner as in the well-known proof of
the one-dimensional Chernoff bound.
\end{proof}

\begin{proposition}
\label{hgratedomprop}
Let $k\ge2$ and let $X_i$, $i\in\set{1,\ldots,k}$, be independent
random variables with distributions {\it Bin}$(n_i,p)$,
respectively. Denote $n=\sum_{i=1}^kn_i$, $X_{1\ldots j}=\sum_{i=1}^jX_i$,
$\bar{X}_i=X_i/n_i$ and $\bar{X}_{1\dots j}=X_{1\dots j}/\sum_{i=1}^jn_i$. Then
\begin{align}
\label{hgratedom1}
\sum_{i=1}^kn_iD_B(\bar{X}_i\|p)-nD_B(\bar{X}_{1\ldots k}\|p)
&\buildrel{(st)}\over\le\sum_{i=1}^{k-1}\left(\log2+Y_i\right),\\
%\label{hgratedom1}
\end{align}
where $Y_1,\ldots,Y_{k-1}$ are independent {\it Exp}$(1)$ random variables.
\end{proposition}
\begin{proof}
For $k=2$, the left hand side of \eqref{hgratedom1} equals
\begin{equation}
\label{hgratedompro1}
n_1D_B(\bar{X}_1\|\bar{X}_{12})+n_2D_B\left(\left.\frac{X_{12}-X_1}{n_2}\right\|\bar{X}_{12}\right)
\end{equation}
by Proposition \ref{binkulidentprop}. For any $N\in\set{0,\ldots,n}$,
consider the conditional distribution of \eqref{hgratedompro1}, given
that $X_{12}=N$. By Proposition \ref{hypergrateprop}, this is the
distribution of the {\it Hypergeometric}$(n,n_1,N)$ rate function
taken at the random variable $X_1$ with the same distribution. The
claim now follows by Lemma \ref{probofinfolemma}, because the
stochastic upper bound does not depend on $N$, i.e.\ on the value of $X_{12}$.

For $k>2$ we proceed by induction. Assume that the claim holds for
$k-1$ and write
\begin{align*}
&\sum_{i=1}^kn_iD_B(\bar{X}_i\|p)-nD_B(\bar{X}_{12}\|p)\\
&=\sum_{i=1}^{k-1}n_iD_B(\bar{X}_i\|p)-(n-n_k)D_B(\bar{X}_{1\ldots(k-1)}\|p)\\
&\quad+n_kD_B(\bar{X}_k\|p)+(n-n_k)D_B(\bar{X}_{1\ldots(k-1)}\|p)-nD_B(\bar{X}_{1\ldots k}\|p).
\end{align*}
By the induction hypothesis, the first row of the second expression is
stochastically bounded by $\sum_{i=1}^{k-2}(\log2+Y_i)$, irrespective
of the value of $X_{1\ldots(k-1)}$. Similarly, the second row is
stochastically bounded by $\log2+Y_k$, where $Y_k\sim\mbox{\it
  Exp}(1)$, irrespective of the value of $X_{1\ldots k}$. It remains
to note that $Y_k$ can be chosen to be independent of
$(Y_1,\ldots,Y_{k-2})$, because $X_k$ is independent of
$(X_1,\ldots,X_{k-1})$, and of $\bar{X}_{1\ldots(k-1)}$ in particular.
\end{proof}

\begin{lem}
\label{binmaxlemma}
Consider the sequence $(G_n,\xi_n)$ of stochastic block models as in
Theorem \ref{blomothm}. Then the following holds.
\begin{enumerate}
\item\label{binmaxclaim}
 For any blocks $A_i$ and $A_j$ such that
$d_{ij}\not\in\set{0,1}$, it holds for an arbitrary $\epsilon>0$ with
high probability that
\begin{equation*}
\min_{v\in A_i}\frac{|e(\set{v},A_j)|}{|A_j|}\ge d_{ij}-n^{-\frac12+\epsilon},\quad
\max_{v\in A_i}\frac{|e(\set{v},A_j)|}{|A_j|}\le d_{ij}+n^{-\frac12+\epsilon}.
\end{equation*}
\item\label{regularityclaim}
For any $\epsilon>0$, the partition $\xi_n$ is
  $\epsilon$-regular with high probability.
\end{enumerate}
\end{lem}
\begin{proof}
Claim \ref{binmaxclaim}:
By Proposition \ref{chernoffprop} and \eqref{infodiffI}, 
\begin{align*}
&\pr{\max_{v\in A_i}\frac{|e(\set{v},A_j)|}{|A_j|}>d_{ij}+h}\\
&\le\sum_{v\in A_i}\pr{\frac{|e(\set{v},A_j)|}{|A_j|}>d_{ij}+h}\\
&\le|A_i|\exp\left(-|A_j|D_B(d_{ij}+h\,\|\,d_{ij})\right)\\
&=|A_i|\exp\left(-|A_j|\left(\frac{h^2}{2d_{ij}(1-d_{ij})}
  +\frac{h^3}{6}D_B'''(z\|d_{ij})\right)\right).
\end{align*}
The last expression converges to zero with the choice
$h=n^{-\frac12+\epsilon}$ (recall that $|A_i|\sim n\gamma_i$ and
$|A_j|\sim n\gamma_j$), which proves the claim on the maximum. The
case of the minimum is symmetric.

Claim \ref{regularityclaim}: Fix $\epsilon>0$ and consider any
$i,j$. Let $U_1\subseteq A_i$ and $U_2\subseteq A_j$ such that
$|U_1|\ge\epsilon|A_i|$ and $|U_2|\ge\epsilon|A_j|$. By Proposition
\ref{chernoffprop},
\begin{equation*}
\pr{|d(U_1,U_2)-d_{ij}|>\epsilon}
\le e^{-|U_1||U_2|D_B(d_{ij}+\epsilon\|d_{ij})}+e^{-|U_1||U_2|D_B(d_{ij}-\epsilon\|d_{ij})}.
\end{equation*}
Let 
$\iota(\epsilon)=\min\set{D_B(d_{ij}+\epsilon\|d_{ij}),D_B(d_{ij}-\epsilon\|d_{ij})}$.
The union bound yields
\begin{align*}
&\pr{\exists U_1\subseteq A_i,\ U_2\subseteq A_j:\
|U_1|\ge\epsilon|A_i|,\ |U_2|\ge\epsilon|A_j|,\ 
|d(U_1,U_2)-d_{ij}|>\epsilon}\\
&\le2|A_i||A_j|\exp\left((|A_i|+|A_j|)\log2
  -\epsilon^2|A_i||A_j|\iota(\epsilon)\right)\\
&\le2n^2\exp\left(n^2\left(\frac{(\gamma_i+\gamma_j)\log2}{n}
  -\gamma_i\gamma_j\epsilon^2\iota(\epsilon)\right)\right)\to0\quad\mbox{as }
  n\to\infty,
\end{align*}
because $\iota(\epsilon)>0$. 
\end{proof}

\begin{remark}
\label{epsilon4rem}
Because $\iota(\epsilon)\propto\epsilon^2$ as $\epsilon\to0$, the
proof of claim \ref{regularityclaim} indicates that with a fixed
$\epsilon$, $\epsilon$-regularity starts to hold when
$n>>\epsilon^{-4}$.
\end{remark}

\subsubsection*{Preliminaries for the Poissonian block model}

By the Poissonian block model, the function $\phi(x)=-x\log x$ replaces binomial entropy in the counterparts of code lengths $L_4+L_5$. We indicate below how the crucial steps of the proofs would change.

Denote by $D_P(b\|a)$ the Kullback-Leibler divergence between
distributions $\mbox{\it Poisson}(a)$ and $\mbox{\it Poisson}(b)$
\begin{equation}
\label{kullpoisdef}
D_P(b\|a)=a-b+b\log b-b\log a.
\end{equation}
For a counterpart to Lemma \ref{infodifflem}, note that, for any $z>0$,
\begin{equation}
\label{phi2der}
\phi''(x)=-\frac1x=-\frac{\D^2}{\D{x^2}}D_P(x\|z).
\end{equation}

%For the Poissonian block model, we note the following facts:

%\begin{lem}
%\label{poissonentropylemma}
%The entropy function $H_P(\alpha)$ is increasing and %concave, and its
%two first derivatives have the expressions
%\begin{align}
%\label{HP'}
%\frac{\D{}}{\D{\alpha}}H_P(\alpha)&=\E{\log\frac{X+1}{\alpha%}},\\
%\label{HP''}
%\frac{\D{}^2}{\D{\alpha^2}}H_P(\alpha)&=
%-\frac{1}{\alpha}+\E{\log\left(1+\frac{1}{X+1}\right)},
%\end{align}
%where $X$ denotes a random variable with distribution %$\mbox{\it
%  Poisson}(\alpha)$.
%\end{lem}

%The proofs for the matrix case and for the Poissonian block %model merge together in the following lemma:

\begin{lem}
\label{poissonHIlemma}
For any $\alpha\in[0,1]$ and $x,y>0$, denote $z=\alpha
x+(1-\alpha)y$. Then we have
\begin{align}
\label{poissonHIeq0}
\phi(z)-(\alpha\phi(x)+(1-\alpha)\phi(y))
&=\alpha D_P(x\|z)+(1-\alpha)D_P(y\|z)\\
\label{poissonHIeq1}
&=zI_{\mathit{Ber}(\alpha)}(\frac{\alpha x}{z}).
\end{align}
\end{lem}
\begin{proof}
The equality \eqref{poissonHIeq0} follows from \eqref{phi2der} similarly as the derivation of \eqref{HI}.
The expression \eqref{poissonHIeq1} is obtained by writing the right
hand side of \eqref{poissonHIeq0} with the substitution $y=(z-\alpha
x)/(1-\alpha)$ and re-combining the log terms.
\end{proof}

The Poissonian counterpart of Proposition \ref{hgratedomprop} is the following.

\begin{proposition}
\label{binratedomprop}
Let $a>0$, $k\ge2$, $n_i\ge1$, $i=1,\ldots,k$, and $n=\sum_in_i$. Let
$X_i$, $i\in\set{1,\ldots,k}$, be independent random variables with
distributions {\it Poisson}$(n_ia)$, respectively. Denote $X_{1\ldots
  j}=\sum_{i=1}^jX_i$, $\bar{X}_i=X_i/n_i$ and $\bar{X}_{1\dots
  j}=X_{1\dots j}/\sum_{i=1}^jn_i$. Then
\begin{align}
\label{poiratedom1}
n\phi(\bar{X}_{1\dots k})-\sum_{i=1}^kn_i\phi(\bar{X}_i)
&\buildrel{(st)}\over\le\sum_{i=1}^{k-1}\left(\log2+Y_i\right),
\end{align}
where $Y_1,\ldots,Y_{k-1}$ are independent {\it Exp}$(1)$ random variables.
\end{proposition}
\begin{proof}
The proof of Proposition \ref{hgratedomprop} can be imitated as follows:
\begin{itemize}
\item Using induction, it suffices to consider the case $k=2$.
\item Apply Lemma \ref{poissonHIlemma} to the left hand side of \eqref{poiratedom1} with $x=\bar{X}_1$, $y=\bar{X}_2$, $z=\bar{X}_{12}$ and $\alpha=n_1/n$. This yields
\begin{equation*}
X_{12}I_{\mathit{Ber}(\frac{n_1}{n})}\left(\frac{X_1}{X_{12}}\right)
=I_{\mathit{Bin}(X_{12},\frac{n_1}{n})}(X_1).
\end{equation*}
\item Now, the conditional distribution of $X_1$ given $X_{12}$ is the
  above binomial distribution. Thus, we can apply Lemma
  \ref{probofinfolemma} in a similar way as in the proof of
  Proposition \ref{hgratedomprop}.
\end{itemize}
\end{proof}

\bigskip
{\bf Acknowledgment.} We thank Dr. Tomi R\"aty for proof-reading the manuscript and for making many useful suggestions. This work was partly financially supported by the Academy of Finland projects 294763 (Stomograph) and 288907.

%\bibliography{szemmdl}
%\bibliographystyle{abbrv}

\end{document}